\documentclass[11pt]{article}
\usepackage[margin=1in]{geometry}
\usepackage{mathpazo}
\usepackage{stmaryrd}
\usepackage[english]{babel}
\usepackage[autostyle, english = american]{csquotes}
\usepackage{anyfontsize}
\usepackage{authblk}

\usepackage{amssymb, amsthm, amsmath, dsfont, mathtools}
\mathtoolsset{showonlyrefs}
\usepackage{latexsym}

\usepackage{eepic}
\usepackage{sectsty}
\usepackage{framed}
\usepackage[shortlabels]{enumitem}
\usepackage{float}
\restylefloat{table}
\usepackage{xifthen,mathrsfs}
\usepackage{ytableau}
\usepackage{subfigure}

\usepackage[pagecolor=white]{pagecolor}
\usepackage{xcolor}
\makeatletter
\newcommand{\globalcolor}[1]{%
  \color{#1}\global\let\default@color\current@color
}
\makeatother

\AtBeginDocument{\globalcolor{black}}

\usepackage{hyperref}
\usepackage{hyperref}
\hypersetup{
    colorlinks = true,
    linkcolor = magenta,
    citecolor = blue,
    urlcolor = blue
}

\hypersetup{pdfpagemode=UseNone}

\newtheorem{theorem}{Theorem}
\numberwithin{theorem}{section}
\newtheorem{lemma}[theorem]{Lemma}
\newtheorem{prop}[theorem]{Proposition}
\newtheorem{cor}[theorem]{Corollary}

\newtheorem{question}[theorem]{Question}

\newtheorem*{theorem*}{Theorem}
\newtheorem*{lemma*}{Lemma}
\newtheorem*{prop*}{Proposition}
\newtheorem*{cor*}{Corollary}
\newtheorem*{obs*}{Observation}
\newtheorem*{conjecture*}{Conjecture}
\newtheorem*{question*}{Question}
\newtheorem*{problem*}{Problem}

\theoremstyle{definition}
\newtheorem{definition}[theorem]{Definition}
\newtheorem{claim}[theorem]{Claim}
\newtheorem{remark}[theorem]{Remark}
\newtheorem{fact}[theorem]{Fact}
\newtheorem{example}[theorem]{Example}

\newcommand{\tinyspace}{\mspace{1mu}}

\newcommand{\tr}{\operatorname{Tr}}

\newcommand{\im}{\operatorname{Im}}
\renewcommand{\ker}{\operatorname{Ker}}

\newcommand{\ip}[2]{\langle #1 , #2\rangle}

\newcommand{\ceil}[1]{\lceil #1 \rceil}

\newcommand{\floor}[1]{\lfloor #1 \rfloor}

\newcommand{\norm}[1]{\lVert\tinyspace #1 \tinyspace\rVert}

\newcommand{\I}{\mathds{1}}

\newcommand{\setft}[1]{\mathrm{#1}}
\newcommand{\Density}{\setft{D}}
\newcommand{\Pos}{\setft{Pos}}

\newcommand{\Unitary}{\setft{U}}
\newcommand{\Herm}{\setft{Herm}}

\newcommand{\Sep}{\setft{Sep}}
\newcommand{\Bisep}{\setft{Bisep}}

\newcommand{\density}[1]{\setft{D}\left(#1\right)}

\newcommand{\complex}{\mathbb{C}}

\renewcommand{\natural}{\mathbb{N}}
\newcommand{\integer}{\mathbb{Z}}

\newcommand\X{\mathcal{X}}
\newcommand\Y{\mathcal{Y}}

\newcommand\A{\mathcal{A}}
\newcommand\B{\mathcal{B}}

\newcommand\U{\mathcal{U}}
\newcommand\C{\mathcal{C}}
\newcommand\D{\mathcal{D}}

\renewcommand\P{\setft{P}}

\renewcommand\H{\mathcal{H}}
\newcommand\J{\mathcal{J}}

\newcommand{\End}{\setft{End}}
\newcommand{\Hom}{\setft{Hom}}

\DeclareMathOperator{\spn}{span}

\newcommand{\w}{{\wedge}}

\newcommand{\lm}{\operatorname{lm}}
\newcommand{\HS}{\setft{HS}}

\newcommand{\ootimes}{ \otimes \cdots \otimes }

\newcommand{\textbigotimes}{{\textstyle \bigotimes}}

\newcommand{\textbigoplus}{{\textstyle \bigoplus}}

\newcommand{\frakS}{\mathfrak{S}}

\newcommand{\bfalpha}{\boldsymbol{\alpha}}
\newcommand{\bfbeta}{\boldsymbol{\beta}}
\newcommand{\bfgamma}{\boldsymbol{\gamma}}

\renewcommand{\emptyset}{\varnothing}
\newcommand{\eps}{\varepsilon}

\newcommand{\under}[1]{\underline{#1}}

\def\ba#1\ea{\begin{align}#1\end{align}}

\newcommand{\varppsi}{\psi\psi^*}

\newcommand{\varpphi}{\phi\phi^*}

\newcommand{\varpp}[2]{(#1_{#2}^{} #1_{#2}^*)}
\newcommand{\varp}[2]{#1_{#2}^{} #1_{#2}^*}

\newcommand{\bull}{\bullet}

\newcommand{\reg}{\setft{reg}}

\newcommand{\hf}{\setft{HF}}
\newcommand{\hs}{\setft{HS}}

\newcommand{\conv}{\setft{conv}}

\newcommand{\MPS}{\setft{MPS}}

\newcommand{\uc}{\under{c}}
\newcommand{\ud}{\under{d}}

\newcommand{\comment}[1]{}

\newcommand{\qip}[1]{}

\begin{document}

\emergencystretch 3em
\title{\bf
$\X$-arability of mixed quantum states
}

\author[$*\dagger$]{Harm Derksen}
 \author[$\ddagger$]{Nathaniel Johnston \thanks{emails: ha.derksen@northeastern.edu, njohnston@mta.ca, benjamin.lovitz@gmail.com}}
 \author[$* \dagger$]{Benjamin Lovitz}
  \affil[$\dagger$]{Department of Mathematics, Northeastern University, Boston, Massachusetts, USA}
  \affil[$\ddagger$]{Department of Mathematics and Computer Science, Mount Allison University,  \authorcr Sackville, New Brunswick, Canada}

\date{\today}

\maketitle
\begin{abstract}
The problem of determining when entanglement is present in a quantum system is one of the most active areas of research in quantum physics. Depending on the setting at hand, different notions of entanglement (or lack thereof) become relevant. Examples include separability (of bosons, fermions, and distinguishable particles), Schmidt number, biseparability, entanglement depth, and bond dimension. In this work, we propose and study a unified notion of separability, which we call~\textit{$\X$-arability}, that captures a wide range of applications including these. For a subset (more specifically, an algebraic variety) of pure states $\X$, we say that a mixed quantum state is \textit{$\X$-arable} if it lies in the convex hull of $\X$.
We develop unified tools and provable guarantees for $\X$-arability, which already give new results for the standard separability problem. Our results include:
\begin{itemize}
\item An \textit{$\X$-tensions} hierarchy of semidefinite programs for $\X$-arability (generalizing the symmetric extensions hierarchy for separability), and a new \textit{de Finetti theorem} for fermionic separability.
\item A hierarchy of eigencomputations for optimizing a Hermitian operator over $\X$, with applications to \textit{$\X$-tanglement witnesses} and polynomial optimization.
\item A hierarchy of linear systems for the \textit{$\X$-tangled subspace} problem, with improved polynomial time guarantees even for the standard entangled subspace problem, in both the generic and worst case settings.
\end{itemize}
%Various notions of ``unentangled states" arrive in different contexts, including separable states of bosons, fermions, or distinguishable particles; states of low Schmidt rank; biseparable states; and matrix product states.
\end{abstract}
\newpage
\tableofcontents
\newpage
%-----------------------------------------------------------------------------%
\section{Introduction}\label{sec:intro}
%-----------------------------------------------------------------------------%

Let $\H=\H_1 \ootimes \H_m$ be a (finite dimensional, complex) Hilbert space over $m$ subsystems. Recall that a density operator (or \textit{mixed state}) $\rho \in \density{\H}$ is~\textit{separable} if it can be written as a probabilistic mixture of pure product states $\rho=\sum_{i} p_i^{} \; \varp{\psi}{i,1} \otimes \dots \otimes \varp{\psi}{i,m}$. Conversely, $\rho$ is \textit{entangled} if it is not separable. Quantum entanglement is one of the central features of modern physics, and the problem of determining when entanglement is present in a quantum system is one of its most active research areas \cite{GT09,HHH09}.

More generally, it is natural to ask if a state $\rho$ can be prepared as a probabilistic mixture of pure states lying in some other set $\X$, in which case we say that $\rho$ is \textit{$\X$-arable}. Conversely, we say that $\rho$ is \textit{$\X$-tangled} if it is not $\X$-arable. We work in the general setting when $\X$ is a \textit{(projective) variety}: the common zero locus of a set of homogeneous polynomials $f_1,\dots, f_p$
\ba\label{eq:x}
\X=\{\varppsi : f_1(\psi)=\cdots = f_p(\psi)=0 \quad \text{and}\quad \norm{\psi}=1\}.
\ea
We say that $f_1,\dots,f_p$ \textit{cut out} $\X$.
\begin{example}\label{ex:varieties} Examples of varieties include: (Pure) product states, bosonic product states, fermionic product states, biseparable states, $\ell$-separable states, $t$-producible states, states of bounded {Schmidt rank}, {matrix product states}, and {tree tensor network states}. 
\end{example}

$\X$-arability thus captures a wealth of applications that are central to quantum entanglement theory, including separability (of both distinguishable and indistinguishable particles)~\cite{HHH09,grabowski2011entanglement}, biseparability~\cite{seevinck2008partial}, $\ell$-separability~\cite{guhne2005multipartite,HHH09}, entanglement depth~\cite{lucke2014detecting}, and Schmidt number~\cite{sanpera2001schmidt}. This motivates the following:

%Three of the most prevalent techniques for studying the separability problem are the \textit{symmetric extensions hierarchy}, \textit{entanglement witnesses}, and \textit{entangled subspaces}. This motivates the following:

%The best line of attack for answering the separability problem depends on the information we are given about $\rho$. If we are handed the (classical) matrix of $\rho$, then the~\textit{symmetric extensions hierarchy} gives a hierarchy of semidefinite programs to determine if $\rho$ is separable. Another tool is the \textit{range criterion}, which says that if $\im(\rho) \subseteq \H$ is an \textit{entangled subspace} (meaning it contains no nonzero product tensors), then $\rho$ is entangled. Alternatively, if one is given the \textit{quantum state} $\rho$, then one can measure it using an \textit{entanglement witness}: a Hermitian operator $H \in \Herm(\H)$ that defines a separating hyperplane for the set of separable states. The symmetric extensions hierarchy can be used to determine if $H$ is an entanglement witness. A commonly used tool to construct entanglement witnesses is also based on entangled subspaces: If $\U \subseteq \H$ is an entangled subspace, then $H=\I_{\H}-\mu \Pi_{\U}$ is an entanglement witness for some $\mu>1$, where $\I_{\H}$ is the identity and $\Pi_{\U}$ is the orthogonal projection onto $\U$.

\begin{question}
Can we develop unified tools and provable guarantees for the $\X$-arability problem?
\end{question}

Three pervasive tools for the separability problem are the \textit{symmetric extensions hierarchy}: A hierarchy of semidefinite programs for deciding if a state is entangled or separable~\cite{DPS04}; \textit{entanglement witnesses}: quantum observables that detect entanglement~\cite{GT09}; and \textit{entangled subspaces}: subspaces avoiding the set of pure product states~\cite{Par04,Bha06}. In answer to our question, we generalize these tools to the $\X$-arability setting, with provable guarantees.

%We obtain the following main results:
%\begin{enumerate}
%\item An \textit{$\X$-tension} hierarchy for $\X$-arability, with applications to Hermitian optimization over $\X$.
%\item A hierarchy of linear systems for the \textit{$\X$-tangled subspace} problem, with polynomial time guarantees.
%\item A hierarchy of eigencomputations for determining the \textit{robustness of $\X$-tanglement} of a subspace, and explicit constructions of robustly entangled subspaces / entanglement witnesses.
%\end{enumerate}

%Entanglement witnesses and entangled subspaces generalize naturally to~\textit{$\X$-tanglement witnesses} (separating hyperplanes for the set of $\X$-arable states) and \textit{$\X$-tangled subspaces} (subspaces avoiding $\X$). This motivates the following questions
%
%\begin{enumerate}
%\item Can we develop an extension hierarchy for $\X$-arability?
%\item Can we develop tools to determine if $H \in \Herm(\H)$ is an $\X$-tanglement witness?
%\item Can we efficiently certify that a subspace $\U \subseteq \H$ is $\X$-tangled?
%\item Can we construct robustly $\X$-tangled subspaces $\U$?
%\end{enumerate}
%
%In this work, we generalize all of these tools to the $\X$-arability problem.
%
%Given a classical description of a state $\rho$, a commonly used tool 
%
%
%Two common tools for determining whether a mixed state is separable or entangled are the~\textit{symmetric extensions hierarchy} and the~\textit{range criterion}. In this work, we generalize these tools to determine $\X$-arability, and prove certain performance guarantees that are new even for separability.

\subsection{$\X$-tensions hierarchy for $\X$-arability}

%Let $I=\{\sum_i f_i g_i : g_i \text{ is a polynomial }\}$ be the \textit{ideal generated by} $f_1,\dots, f_p$, and let $I_k \subseteq I$ be the homogeneous polynomials of degree $k$.

%The \textit{symmetric extensions hierarchy} says that a state $\rho \in \Density(\H)$ is separable if and only if for all $k$ it admits an extension whose image (or range) is contained in the symmetric space 

We obtain the following semidefinite programming hierarchy for deciding if a state is $\X$-arable (see Section~\ref{sec:xarable}):

%In a simplified form, our $\X$-tension hierarchy for $\X$-arability says the following (see Section~\ref{sec:xarable}):

\begin{theorem}[$\X$-tensions hierarchy]\label{thm:intro_xtension}
Let $\rho \in \Density(\H)$ be a state. Then $\rho$ is $\X$-arable if and only if for all $k$ there exists $\sigma_k \in \Density(\H^{\otimes k})$ for which $\tr_{k-1}(\sigma_k)=\rho$ and $\im(\sigma_k) \subseteq \X^k$, where
\ba\label{eq:im}
\X^k:=\spn\{\psi^{\otimes k} : \varppsi \in \X\}.
\ea
\end{theorem}
%One can take $I_k^{\perp}=\spn\{\psi^{\otimes k} : \varppsi \in \X\}$.

We call $\sigma_k$ an \textit{$\X$-tension} of $\rho$.  When $\X= \X_{\Sep}$ is the set of pure product states, this specializes to the symmetric extensions hierarchy of~\cite{DPS04} with all $m$ subsystems extended simultaneously. More generally, one can replace $\X^k$ with $I_k^{\perp}$, where $I$ is the ideal generated by any $f_1,\dots, f_p$ cutting out $\X$ (see Section~\ref{sec:intro_null} below). While it is easy to describe $\X^k$ (or $I_k^{\perp}$) abstractly as in~\eqref{eq:im}, it can be difficult to write down an explicit basis. We give explicit descriptions of $\X^k$ (or $I_k^{\perp}$) for all of the varieties in Example~\ref{ex:varieties}.

In the case of {fermionic separability} we prove a \textit{quantitative} version of this result, which we call a \textit{de Finetti theorem for fermionic separability} (see Section~\ref{sec:qdf}). The key ingredient is the observation that $\X^k$ forms an irreducible representation of the local unitary group in this case, allowing us to apply the results of~\cite{koenig2009most}. Related results are proven in~\cite{crismale2012finetti,kraus2013ground,krumnow2017fermionic}, although these use a mode-partitioned notion of fermionic separability that is closer to the distinguishable-particle setting. By contrast, our result applies to the particle-level notion of fermionic separability introduced in~\cite{grabowski2011entanglement} (see Example~\ref{ex:fermion}). These and other notions of entanglement in fermionic systems have found widespread applications~\cite{benatti2020entanglement}.

We illustrate the example of Schmidt number. Let $\H=\H_1 \otimes \H_2$ and let $\X_r \subseteq \P(\H)$ be the set of pure states of Schmidt rank at most $r$. The least $r$ for which $\rho$ is $\X_r$-arable is called the \textit{Schmidt number} of $\rho$. In the following, let $\Pi_{i,r+1}^{\w}$ be the projection onto the antisymmetric subspace of $\H_i^{\otimes r+1}$, and let $S^k(\H)\subseteq \H^{\otimes k}$ be the symmetric subspace.
\begin{cor}
$\rho\in \Density(\H)$ is $\X_r$-arable if and only if for all $k \geq r+1$ there exists $\sigma_k \in \Density(\H^{\otimes k})$ for which $\tr_{k-1}(\sigma_k)=\rho$ and $\im(\sigma_k) \subseteq \X_r^k$, where
\ba
\X_r^k = \ker(\Pi_{1,r+1}^{\w} \otimes \Pi_{2,r+1}^{\w} \otimes \I_{\H}^{\otimes k-r-1}) \cap S^k(\H).
\ea
\end{cor}

We note that the hierarchy in Theorem~\ref{thm:intro_xtension} (and in this corollary) can be strengthened at each level by additionally imposing that $\sigma_k$ is has positive partial transpose (PPT) with respect to some bipartitions of $\H^{\otimes k}$ (see Remark~\ref{rmk:ppt}). Similar statements apply to the other hierarchies below.

\subsection{Hermitian optimization over $\X$}

Using Theorem~\ref{thm:intro_xtension}, we can obtain hierarchies for constrained Hermitian optimization problems of the form

\ba\label{eq:intro_opt}
H_{\X}^{\downarrow}:=\min_{\psi \psi^* \in \X} \ip{\psi}{H \psi} \quad\quad \text{or}\quad\quad H_{\X}^{\uparrow}:=\max_{\psi \psi^* \in \X} \ip{\psi}{H \psi},
\ea
where $H\in \Herm(\H)$ is any Hermitian operator. Even in the special case when $\X=\X_{\Sep}$, this problem has found many applications including computing the geometric measure of entanglement, determining the performance of QMA(2) protocols, and determining the ground-state energy of mean-field Hamiltonians. Moreover, Ref. \cite{HM10} contains a list of 21 equivalent or closely related problems in quantum information and theoretical computer science. As further motivation, $H$ forms a separating hyperplane for the set of $\X$-arable states (which we call an \textit{$\X$-tanglement witness}) if and only if $H_{\X}^{\downarrow} \geq 0$ and $H$ has at least one negative eigenvalue.

We use Theorem~\ref{thm:intro_xtension} to obtain a hierarchy of eigencomputations for the optimization problem~\eqref{eq:intro_opt}. In the following, let $\lambda_{\min}(\cdot)$ be the minimum eigenvalue.
\begin{cor}\label{cor:intro_eigen}
Let $\Pi_{\X^k}$ be the orthogonal projection onto $\X^k$, and let $\nu_k=\lambda_{\min}(\Pi_{\X^k} (H \otimes \I^{\otimes k-1}) \Pi_{\X^k})$. Then $\nu_1 \leq \nu_2 \leq \dots$, and $\lim_k \nu_k = H_{\X}^{\downarrow}$.
\end{cor}
In particular, this gives a hierarchy of eigencomputations for determining if $H$ is an $\X$-tanglement witness. We make this hierarchy explicit for all of the varieties mentioned in Example~\ref{ex:varieties}.

Corollary~\ref{cor:intro_eigen} can also be proven as a consequence of~\cite[Theorem 1]{catlin1999isometric} if one is familiar with vector bundles. In Section~\ref{sec:poly} we apply Corollary~\ref{cor:intro_eigen} to constrained Hermitian {polynomial} optimization. In particular, we prove that the Hermitian sum of squares hierarchy of~\cite{d2009polynomial}, which is based on semidefinite programming, is equivalent to a hierarchy of \textit{eigencomputations} in some settings, which can lead to computational savings.
%  relate this hierarchy to the \textit{Hermitian sum of squares} (HSOS) hierarchy~\cite{d2009polynomial}, a hierarchy of semidefinite programs for Hermitian {polynomial} optimization. In particular, we prove that the HSOS hierarchy is equivalent to a hierarchy of \textit{eigencomputations} in some settings, which can lead to computational savings.

%
%Surprisingly, we prove that the two hierarchies are equivalent in some settings, so the HSOS hierarchy can be recast as a hieararchy of eigencomputations.
%
%roving that HSOS is equivalent to a hierarchy of eigencomputations in some settings. This can potentially lead to computational savings, as eigencomputations are more practical than full semidefinite programs (see ).

%By Corollary~\ref{cor:intro_eigen}, we can also obtain an explicit hierarchy of eigencomputations for optimizing a Hermitian operator over $\X_r$. See Table~\ref{table:xarable} for many more examples.

\subsection{$\X$-tangled subspaces}\label{sec:intro_null}

\textit{Entangled subspaces} are subspaces of $\H$ that exhibit some preset notion of entanglement~\cite{GW07}. We define an \textit{$\X$-tangled subspace} to be a subspace $\U \subseteq \H$ that does not intersect $\X$. This captures many well-studied notions of entangled subspaces, including completely entangled subspaces, genuinely entangled subspaces, subspaces of high {entanglement depth},  $r$-entangled subspaces, and their bosonic/fermionic variants~\cite{Par04,CMW08,LJ21,demianowicz2024completely}. An immediate application of $\X$-tangled subspaces is that any mixed quantum state supported on an $\X$-tangled subspace is $\X$-tangled. $\X$-tangled subspaces can also be used to construct $\X$-tanglement witnesses~\cite{Hor97,BDMSST99}. Numerous other applications have appeared in recent years, including quantum error correction \cite{GW07,HG20} and quantum cryptography~\cite{SS19}.

More generally, different~\textit{measures} of entanglement of subspaces have found applications in entanglement theory and quantum communication theory, most notably Hastings' disproof of additivity of the Holevo capacity~\cite{hayden2006aspects,hastings2009superadditivity,aubrun2011hastings,demianowicz2019entanglement,zhu2024quantifying}. We introduce and study the \textit{geometric measure of $\X$-tanglement} (GM$\X$) of $\U$, given by
\ba
E_{\X}(\U):=1-(\Pi_{\U})_{\X}^{\uparrow}\;,
\ea
where $\Pi_{\U}$ denotes the orthogonal projection onto $\U$ and $(\Pi_{\U})_{\X}^{\uparrow}$ is defined in~\eqref{eq:intro_opt}. Note that ($E_{\X}(\U)>0 \iff \U$ is $\X$-tangled). This generalizes the well-studied geometric measure of entanglement for pure states~\cite{wei2003geometric} as well as notions for subspaces studied in~\cite{demianowicz2019entanglement,zhu2024quantifying}. %In Section~\ref{sec:xtangled_sub} we observe that subspaces with high geometric measure can be used to construct $\X$-tanglement witnesses with highly negative eigenvalues, and also to certify robustly $\X$-tangled mixed quantum states (states that remain $\X$-tangled even after small unitary perturbations; see also~\cite{zhu2024quantifying}).
%In~\cite{JLV22a} a hierarchy of linear systems was introduced to certify entanglement of subspaces in several cases. The following result generalizes the hierarchy of~\cite{JLV22a} to a general variety $\X$, and moreover to determine the geometric measure of $\X$-tanglement.
An immediate consequence of Corollary~\ref{cor:intro_eigen} is a hierarchy of eigencomputations for computing $E_{\X}(\U)$.
\begin{cor}\label{cor:intro_subspace}
For each positive integer $k$, let $\nu_k$ be the maximum eigenvalue of
\ba
\Pi_{\X^k} (\Pi_{\U} \otimes \I_{\H}^{\otimes k-1}) \Pi_{\X^k}.
\ea
Then $\nu_1 \geq \nu_2 \geq \dots$ and $E_\X(\U)= 1-\lim_{k\rightarrow \infty} \nu_k$.
\end{cor}
We make this hierarchy explicit for all of the examples mentioned in Example~\ref{ex:varieties}. Let us focus on the special case of this hierarchy that simply checks if $\U$ is $\X$-tangled or not (see Section~\ref{sec:equiv}):
\begin{cor}[Nullstellensatz hierarchy]
$\U$ is $\X$-tangled if and only if $\X^k \cap \U^{\otimes k}=\{0\}$ for some $k$.
\end{cor}

% Our hierarchy specializes to say that $\U$ is $\X$-tangled if and only if $\nu_k <1$ for some $k$. This is in fact equivalent to a hierarchy of linear systems which checks if $\X^k \cap \U^{\otimes k} = \{0\}$ at each level; reproducing the hierarchy introduced in~\cite{JLV22a} (see Section~\ref{sec:equiv} for details).
We call this the \textit{Nullstellensatz hierarchy} because it can also be derived from Hilbert's Nullstellensatz. For special cases of varieties $\X$, this specializes to the hierarchies studied in~\cite{JLV22a}. We generalize and improve upon the polynomial time guarantees given in~\cite{JLV22a} for this hierarchy, both in the generic and worst case settings.

% ($\nu_k <1 \iff \X^k \cap \U^{\otimes k} = \{0\}$), giving a hierarchy of linear systems that specializes to the

% hierarchy of~\cite{JLV22a} (see Section~\ref{sec:equiv} for details). We call this the \textit{Nullstellensatz hierarchy} because it can be derived from Hilbert's Nullstellensatz.

\subsubsection{Generic degree bounds for the Nullstellensatz hierarchy}
While the $\X$-tangled subspace problem is NP Hard in the worst case, we prove polynomial time guarantees when $\U$ is generically chosen (or ``typical"). The following theorem shows that the Nullstellensatz hierarchy certifies $\X$-tanglement of generically chosen subspaces in polynomial time, up to an arbitrarily small multiplicative loss in dimension.

%We prove polynomial time guarantees for the Nullstellensatz hierarchy, which generalize and improve upon the results of~\cite{JLV22a} even for the standard entangled subspace problem with $\X=\X_{\Sep}$. The following theorem shows that the Nullstellensatz hierarchy certifies $\X$-tanglement of generically chosen (or ``typical") subspaces in polynomial time, up to an arbitrarily small multiplicative loss in dimension.

%This theorem specializes to prove that generic subspaces of dimension $(1-\eps)N$ are certified $\X$-tangled at a constant level of the Nullstellensatz hierarchy (which takes polynomial time to check).
%The following corollary applies to all of the varieties in Example~\ref{ex:varieties}, and even more broadly to projected entangled pair states (PEPS) and arbitrary tensor network states:

%Let us focus on the special case of our hierarchy that simply checks if the linear system $\X^k \cap \U^{\otimes k} = \{0\}$ holds in order to certify $\X$-tanglement of $\U$. We call this the \textit{Nullstellensatz hierarchy} because it can be derived from Hilbert's Nullstellensatz. We prove polynomial time guarantees for this hierarchy which generalize and improve upon the results of~\cite{JLV22a}, even for the standard entangled subspace problem with $\X=\X_{\Sep}$.

\begin{theorem}\label{thm:poly}
Let $\X$ be any of the varieties listed in Example~\ref{ex:varieties}, let $0 < \eps < 1$ be arbitrary, and let $N=\dim(\H)$ be sufficiently large. Then a generically chosen subspace $\U \subseteq \H$ of dimension $\dim(\U)=(1-\eps) N$ is $\X$-tangled, and this is certified by the Nullstellensatz hierarchy in polynomial time.
\end{theorem}
This improves the genericity guarantees of~\cite{JLV22a}, which hold only for particular choices of $\eps$. For example, if $\X=\X_1$ is the set of pure bipartite product states, then the cited work only applies when $\eps\geq 3/4$.  Our improvement comes as a result of using more sophisticated algebraic techniques. We note that the work~\cite{JLV23a} proposes a related algorithm to \textit{recover} elements of $\U$ contained in $\X$ (see Remark~\ref{rmk:related}).

We prove Theorem~\ref{thm:poly} by showing that a generically chosen subspace $\U$ can be certified $\X$-tangled at a constant level (or \textit{degree}) $k$ of the Nullstellensatz hierarchy. Theorem~\ref{thm:poly} is derived from the following general-purpose bound proven in Section~\ref{sec:null}:

\begin{theorem}[Generic degree bound]\label{thm:generic_intro}
Let $N=\dim(\H)$, and let $s$ and $k$ be positive integers. If
\ba
\dim(\X^k) < \binom{N-s+k}{k},
\ea
then a generically chosen $s$-dimensional subspace $\U \subseteq \H$ is $\X$-tangled, and satisfies $\X^k \cap \U^{\otimes k} = \{0\}$.
\end{theorem}

\subsubsection{Worst case degree bounds for the Nullstellensatz hierarchy}
We also analyze the worst-case performance of the Nullstellensatz hierarchy. Remarkably, the Nullstellensatz hierarchy is guaranteed to terminate (i.e. detect every $\X$-tangled subspace) at a finite degree $k$; something that is known not to be possible for separability hierarchies such as symmetric extensions~\cite{fawzi2021set}. Moreover, we use algebraic-geometric techniques to give explicit upper bounds on (or even determine exactly) the worst-case degree $k$ for all of the varieties in Example~\ref{ex:varieties} (see Table~\ref{table:nullstellensatz}).

Let us start with the example $\X_r \subseteq \P(\complex^{n_1}\otimes \complex^{n_2})$; the set of pure states of Schmidt rank at most $r$. Surprisingly, our results show that the $\X_r$-tangled subspace problem has a worst case polynomial time algorithm when $n_1$ (or $n_2$) are fixed:
\begin{theorem}\label{thm:worst_poly}
The worst case degree required by the Nullstellensatz hierarchy to certify $\X_r$-tanglement is precisely $k=r(\min\{n_1,n_2\}-r)+1$. In particular, the Nullstellensatz hierarchy gives a worst case polynomial time algorithm for the $\X_r$-tangled subspace problem when $n_1$ (or $n_2$) are fixed.
\end{theorem}

For example, consider the specific case of checking entanglement of a qubit--qudit subspace, i.e. $r=1, n_1=2$. Then this theorem shows that the $k=2$ level of the hierarchy is all that is needed, so entanglement of a subspace in this case can be determined by solving a $\binom{n_2}{2} \times \binom{\dim \U + 1}{2}$ linear system (see Section~\ref{sec:equiv}). Furthermore, code that does this is provided in~\cite{JLV22a} (although it was not known at that time that the code gave an exact answer when $k=2$ in this case).

%\jnote{Is it worth noting the even more specific case of checking entanglement of a qubit--qudit subspace? That is, $r = 1$, $n_1 = 2$. Then the $k = 2$ level of the hierarchy is all that is needed, so entanglement of a subspace in this case can be determined by solving a $n_2^2(n_2-1) \times \binom{\mathrm{dim}(U) + 2}{3}$ linear system. Maybe even mention that code that does this is provided by our old paper (but it wasn't know at that time that the code gave an exact answer when $k = 2$ in this case)?}

To describe our worst-case bounds more generally, it will be convenient to describe the Nullstellensatz hierarchy in greater generality. Let $I=\{\sum_i f_i g_i : g_i \text{ is a polynomial}\}$ be the ideal generated by $f_1,\dots, f_p$, let $R_k$ be the set of homogeneous polynomials on $\H$ of degree $k$, and let $I_k = I \cap R_k$. Similarly, let $I(\U)$ be the ideal generated by the linear equations defining $\U$. The $k$-th degree of the \textit{Nullstellensatz hierarchy} checks if $I_k+I(\U)_k=R_k$. If this equality holds, then $\U$ is $\X$-tangled. If $I=I(\X)$ is the set of \textit{all} polynomials vanishing on $\X$ (which is perhaps more than enough to simply \textit{cut out} $\X$), then this equality is equivalent to $\X^k \cap \U^{\otimes k} = \{0\}$, reproducing the simplified hierarchy described above (see Section~\ref{sec:equiv}).

To state our worst-case degree bounds, we require the technical notions of \textit{Cohen-Macaulayness} of $R/I$ and the \textit{(Castelnuovo-Mumford) regularity} $\reg(R/I)$, which are defined in Section~\ref{sec:null}.
\begin{theorem}[Worst case degree bound]\label{thm:intro_worst_case}
Let $N=\dim(\H)$, and let $\X$ be a variety cut out by homogeneous polynomials $f_1,\dots, f_p$ of degree at most $d$. Then degree $k=N(d-1)+1$ of the Nullstellensatz hierarchy suffices to certify any $\X$-tangled subspace. If $R/I$ is \textit{Cohen-Macaulay}, then the worst-case degree is precisely $\reg(R/I)+1$.
\end{theorem}

This theorem applies to all of the varieties mentioned in Example~\ref{ex:varieties}; see Table~\ref{table:nullstellensatz}. Furthermore, many of these varieties satisfy the Cohen-Macaulay condition, allowing us to precisely determine the worst-case degree required. Since the $\X$-tangled subspace problem is NP hard, we can expect the worst-case degree $k$ to be at least linear in $N=\dim(\H)$ in general~\cite{barak2017quantum}. Conversely, the degree bound $N(d-1)+1$ appearing in Theorem~\ref{thm:intro_worst_case} shows that linear scaling is often sufficient, since $d$ is constant in $N$ for many varieties of interest.

\section{Background and notation}\label{sec:background}

In this section we review some background and notation for this work. See e.g. \cite{NC00,wilde2013quantum,Wat18} for more details on quantum information theory, and \cite{shafarevich1988basic,landsberg2012tensors,harris2013algebraic} for more details on varieties and the symmetric algebra. We work coordinate-independently; see and~\cite{hoffmann1971linear,artin2011algebra,halmos2017finite} for background on abstract algebra.

Let $\H_1,\dots, \H_m$ be (finite dimensional, complex) Hilbert spaces, and let $\H=\H_1\otimes \dots \otimes \H_m$ be the Hilbert space with Hermitian form $\ip{\cdot}{\cdot}$ induced from the Hermitian forms on $\H_i$, which we take to be antilinear in the first argument. Let $\H^*$ be the dual space. The Hermitian inner product defines an antilinear isomorphism $\H \cong \H^*$ given by $\psi^*=\ip{\psi}{-}$. Let $(\cdot,\cdot): \H^* \times \H \rightarrow \complex$ be the bilinear form $(f,\psi)=f(\psi)$. Let $N=\dim(\H)$, and let $e_1,\dots, e_N \in \H$ be an orthonormal basis with dual basis $x_1,\dots, x_N \in \H^*$.

For a Hilbert space $\J$, let $\Hom(\J,\H)$ be the set of linear maps (homomorphisms) from $\J$ to $\H$, Let $\End(\H)=\Hom(\H,\H)$ be the set of endomorphisms of $H$, let $\Unitary(\H)$ be the set of unitary operators, let $\Herm(\H)\subseteq \End(\H)$ be the set of Hermitian operators, let $\Pos(\H)\subseteq \Herm(\H)$ be the set of positive semidefinite operators, and let $\density{\H}\subseteq \Pos(\H)$ be the set of \textit{density operators} (or \textit{mixed states}, or simply \textit{states}): positive semidefinite operators of trace one. Let $\P(\H) \subseteq \Density(\H)$ be the set of \textit{pure states}: rank-one states, i.e. states that can be written as $\psi \psi^*$ for a unit vector $\psi\in \H$.

For a subset $S \subseteq [m]$, let $\tr_S : \End(\H) \rightarrow \End(\bigotimes_{i \notin S} \H_i)$ be the partial trace $\tr_S=\I_{\otimes_{i \notin S} \H_i} \otimes \tr_{\otimes_{i\in S} \H_i}$. For an integer $i \in [d]$ let $\tr_i : \End(S^d(\H)) \rightarrow \End(S^{d-i}(\H))$ be the partial trace over any $i$ copies of $\H$ (it does not matter which).

Let $\lambda_{\min}(\cdot)$ and $\lambda_{\max}(\cdot)$ be the minimum and maximum eigenvalues, respectively. For an integer tuple $\bfalpha=(\alpha_1,\dots, \alpha_d),$ let $|\bfalpha|=\alpha_1+\dots+\alpha_d$.

\subsection{Symmetric algebra, ideals, and varieties}

For a permutation $\sigma \in \frakS_d$, let $U_{\sigma} \in \Unitary(\H^{\otimes d})$ be the corresponding permutation of tensor factors $U_{\sigma} ( \psi_1 \otimes \dots \otimes \psi_d)=\psi_{\sigma^{-1}(1)}\otimes \dots \otimes \psi_{\sigma^{-1}(d)}$, extended linearly. The \textit{symmetric subspace} $S^d(\H)\subseteq \H^{\otimes d}$ is the subspace of vectors $v$ for which $U_{\sigma} v = v$ for all $\sigma \in \frakS_d$. Let $\Pi_d=\frac{1}{d!} \sum_{\sigma \in \frakS_d} U_{\sigma}$ be the orthogonal projection onto the symmetric subspace. The \textit{antisymmetric subspace} $\Lambda^d(\H)\subseteq \H^{\otimes d}$ is the set of vectors $v$ for which $U_{\sigma} v = \setft{sign}(\sigma) v$ for all $\sigma \in \frakS_d$. The antisymmetric subspace is spanned by vectors of the form $v_1 \w \cdots \w v_d := \frac{1}{d!} \sum_{\sigma \in \frakS_d} \setft{sgn}(\sigma) v_{\sigma(1)} \ootimes v_{\sigma(d)}$.

For $f\in S^d(\H^*)$ and $\psi \in \H$ we use the shorthand $f(\psi) := (f,\psi^{\otimes d})$. For $g \in S^c(\H^*)$ we define $f \cdot g := \Pi_{c+d} (f\otimes g) \in S^{c+d}(\H^*)$, where $\Pi_{c+d}$ is the projection onto $S^{c+d}(\H^*)$. This makes $S^{\bullet}(\H^*):= \bigoplus_{d=0}^{\infty} S^d(\H^*)$ into an algebra, called the \textit{symmetric algebra}. Let $S^{\ud} (\H^*)= \bigoplus_{c=0}^{d} S^c(\H^*)$.

 Note that $S^{d}(\H^*)$ is isomorphic to the space $\complex[\H]_{d}$ of homogeneous degree $d$ polynomials on $\H$, by the map which sends $\ell^{\otimes d}$ to $\ell^{d}$ for $\ell \in \H^*$, extended linearly. Under this isomorphism, $f \cdot g$ corresponds to the product of polynomials. For this reason, we often refer to elements of $S^{\ud}(\H^*)$ as \textit{polynomials} and elements of $S^d(\H^*)$ as \textit{homogeneous polynomials} (or \textit{forms}) of degree $d$.

%In coordinates, this map acts by
%\ba
%\sum_{\substack{\bfalpha \in \integer_{\geq 0}^n \\ \abs{\bfalpha}=d}} c_{\bfalpha} \binom{d}{\alpha_1,\dots, \alpha_n} x_1^{\alpha_1} \cdots x_n^{\alpha_n} \mapsto \frac{1}{d!} \sum_{\substack{\bfalpha \in \integer_{\geq 0}^n \\ \abs{\bfalpha}=d}} c_{\bfalpha} \sum_{\sigma \in \frakS_d} \sigma \cdot (x_1^{\otimes \alpha_1} \ootimes x_n^{\otimes \alpha_n}).
%\ea

%Let $S^{\under{d}}(\H^*)=\bigoplus_{c=0}^d S^c(\H^*)$. Let $S^{\bullet}(\H^*):=\textbigoplus_{d=0}^{\infty} S^d(\H^*)$ be the \textit{symmetric algebra}, which is identified with the polynomial algebra $\complex[\H]$. One multiplies two elements $f,g \in S^{\under{d}}(\H^*)$ by regarding them as polynomials, or equivalently by projecting $f \otimes g$ onto the subspace $S^{\under{2d}}(\H^*)\subseteq S^{\under{d}}(\H^*)\otimes S^{\under{d}}(\H^*)$. An element $f \in S^{d}(\H^*)$ is evaluated at $\psi \in \H$ by $f(\psi) = (f,\psi^{\otimes d})$. For a vector $v \in S^d(\H)$ and a form $\theta \in \H^*$, let $\theta \hook v := (\theta \otimes \I^{\otimes d-1})v$ be the contraction of $v$ by $\theta$.

An \textit{ideal} $I\subseteq S^{\bull}(\H^*)$ is a linear subspace for which $S^{\bull}(\H^*) \cdot I \subseteq I$. For a subset $J \subseteq S^{\bull}(\H^*)$, let $\langle J \rangle := \{\sum_i f_i g_i : f_i \in S^{\bull}(\H^*), g_i \in J\}$ be the ideal generated by $J$. By Hilbert's basis theorem, every ideal is generated by finitely many polynomials $I=\langle f_1,\dots, f_p\rangle$. An ideal $I$ is \textit{homogeneous} if it can be generated by (finitely many) homogeneous polynomials $f_i \in S^{d_i}(\H^*)$. We say that $I$ is \textit{generated in degree (at most) $d$} if one can take $d_i = d$ (or $d_i \leq d$). It is a standard fact that if $I$ is homogeneous then it can be written as $I=\bigoplus_{d=0}^{\infty} I_d$, with each $I_d \subseteq S^d(\H^*)$ a linear subspace called the \textit{degree-$d$ component} of $I$. For a positive integer $c$ let $I_{\underline{c}}=\bigoplus_{d=0}^c I_d$.

For a subspace $\U \subseteq \H$, let $\U^{\perp} \subseteq \H^*$ be the orthogonal complement of $\U$ with respect to the bilinear form $(\cdot,\cdot)$, and let $\Pi_{\U} \in \Pos(\H)$ be the orthogonal projection (with respect to $\ip{\cdot}{\cdot}$) onto $\U$. For a homogeneous ideal $I\subseteq S^{\bull}(\H^*)$ with degree-$d$ component $I_d \subseteq S^d(\H^*)$, we will often consider $I_d^{\perp}\subseteq S^d(\H)$. Let $\Pi_{I,d}:= \Pi_{I_d^{\perp}} \in \Pos(S^d(\H))$, and let $\Pi_d:=\Pi_{S^d(\H)} \in \Pos(\H^{\otimes d})$. We will often view $\Pi_{I,d}$ as an element of $\Pos(\H^{\otimes d})$ by setting it to zero on the orthogonal complement to $S^d(\H)$ in $\H^{\otimes d}$.

For a homogeneous ideal $I$, let $V(I)\subseteq P(\H)$ be the set of pure states $\psi \psi^*$ for which $f(\psi)=0$ for all $f \in I$ (homogeneity ensures that this condition does not depend on phase). A \textit{(projective) variety} $\X \subseteq \P(\H)$ is a subset of the form $\X=V(I)$ for a homogeneous ideal $I$. One then says that $I$ \textit{cuts out} $\X$. The \textit{ideal of $\X$}, denoted $I(\X)$, is the set of all polynomials vanishing on $\X$.
% The \textit{radical} of $I$ is defined as $\sqrt{I}=\{f : f^k \in I \text{ for some k}\}$. By Hilbert's nullstellensatz, $I(\X)=\sqrt{I}$ for any ideal $I$ cutting out $\X$.
A special case of Hilbert's Nullstellensatz is the following:
\begin{theorem}[Hilbert's weak nullstellensatz]\label{thm:hilbert}
Let $I \subseteq S^{\bull}(\H^*)$ be a homogeneous ideal. Then ${V(I)=\emptyset}$ if and only if $I_k=S^k(\H^*)$ for $k \gg 0$.
\end{theorem}

% the local unitary group $\Unitary(\H_1) \times \dots \times \Unitary(\H_m)$.

%Let $\density{\H}$ be the set of density operators on $\H$, and let $P(\H) \subseteq \density{\H}$ be the set of pure states.

%For a positive integer $d$, let $S^d(\H)\subseteq \H^{\otimes d}$ be the \textit{symmetric subspace}: vectors $v=(v_{i_1,\dots, i_d})$ for which $v$

% We identify the polynomial ring $\complex[\H]$ with the 

\section{$\X$-arability and $\X$-tendability}\label{sec:xarable}

Let $I \subseteq S^{\bull}(\H^*)$ be a homogeneous ideal and let $\X=V(I) \subseteq \P(\H)$ the associated (projective) variety. We say that a state $\rho \in \density{\H}$ is \textit{$\X$-arable} if $\rho\in \conv (\X)$, and otherwise we that $\rho$ is \textit{$\X$-tangled}. Our main result in this section is an \textit{$\X$-tensions} hierarchy for $\X$-arability, which generalizes the well-known symmetric extensions hierarchy for separability~\cite{DPS04}. We use this result to obtain explicit hierarchies for all of the varieties in Example~\ref{ex:varieties}. In the case of fermionic separability, we also prove a \textit{quantum de Finetti theorem}, giving quantitative convergence guarantees for this hierarchy.

Recall that $I_k^{\perp} \subseteq S^k(\H)$ is the orthogonal complement to $I_k \subseteq S^k(\H^*)$ with respect to the bilinear pairing (see Section~\ref{sec:background} for more details).
\begin{definition}
For a state $\rho \in \Density(\H)$, a \textit{$(k,\X)$-tension} of $\rho$ is a state $\sigma_k\in \Density(\H^{\otimes k})$ for which the following two properties hold:
\begin{enumerate}
\item $\setft{Im}(\sigma_k) \subseteq I_k^{\perp}$, and
\item $\tr_{k-1}(\sigma_k)=\rho$.
\end{enumerate}
We say $\rho$ is \textit{$(k,\X)$-tendable} if there exists a $(k,\X)$-tension of $\rho$.
\end{definition}
This is a slight abuse of notation, as a $(k,\X)$-tension depends on the ideal $I$ that one chooses to cut out $\X$.
%The set of $\X$-arable states forms a closed, convex cone, so any $\X$-tangled state admits a separating hyperplane from the set of $\X$-arable states:
%
%\begin{prop}\label{prop:xpositive}
%A state $\rho \in \density{\H}$ is $\X$-arable if and only if $\tr(\rho H) \geq 0$ for all $\X$-positive $H \in \Herm(\H)$.
%\end{prop}
%
%Combining this proposition with Theorem~\ref{thm:robust_hier}, we obtain the following hierarchy for $\X$-arability

\begin{theorem}[$\X$-tensions hierarchy for $\X$-arability]\label{thm:xtension}
Let $\rho \in \density{\H}$ be a state. Then $\rho$ is $\X$-arable if and only if $\rho$ is $(k,\X)$-tendable for all $k$.
\end{theorem}
\begin{proof}
If $\rho$ is $\X$-arable, then $\rho=\sum_i p_i \varp{\psi}{i}$ for some $\varp{\psi}{i} \in \X$ and probability vector $p$. Then $\rho$ is $(k,\X)$-tendable to $\sigma_k := \sum_i p_i \varpp{\psi}{i}^{\otimes k}$.

Conversely, let $I$ be generated in degree at most $d$, and for each $k \geq d$ let $\tau_k = \tr_{k-d}(\sigma_k)$. By compactness of the set of density matrices, we can find a subsequence $\{\tau_{k_j}\}$ so that $\lim_{j \rightarrow \infty} \tau_{k_j}$ exists (call it $\tau$). Since $\im(\sigma_k) \subseteq S^k(\H)$, $\tau$ has a symmetric extension to $k_j$ copies for arbitrarily large values of $k_j$, so $\tau \in \conv\{(\varppsi)^{\otimes d} : \psi \in \H\}$ by~\cite{DPS04} (or the bosonic quantum de Finetti theorem reproduced in Theorem~\ref{thm:bosonicdefinetti} below). Let $\tau=\sum_i p_i \varpp{\psi}{i}^{\otimes d}$ with each $p_i >0$. Since $\im(\tau)\subseteq I_d^{\perp}$, it follows that $\psi_i^{\otimes d} \in I_d^{\perp}$ for all $i$, so $\varp{\psi}{i} \in \X$ for all $i$. Hence, $\rho= \tr_{d-1}(\tau)=\sum_i p_i \varp{\psi}{i}$ is $\X$-arable.
%
%
%
%By Theorem~\ref{thm:robust_hier}, $H$ is $\X$-positive if and only if the minimum eigenvalue $\nu_k$ of
%\ba
%\Pi_{I,k} (H \otimes \I_{\H}^{\otimes k-1}) \Pi_{I,k}
%\ea
%is non-negative for every positive integer $k$. Combined with Proposition~\ref{prop:xpositive}, this shows that $\rho$ is $\X$-arable if and only if the solution to the following SDP is non-negative for all $k$:
%
%
%\begin{align}\begin{split}\label{eq:sep_sdp_dual}
%        \textup{minimize:} & \ \tr(\rho H)\\
%        % & \ \tr\big((W \otimes I_{\V}^{\otimes (d-1)}) \tr_{k-d}(\rho)\big) \\
%        \textup{subject to:} & \ \Pi_{I,k} (H \otimes \I_{\H}^{\otimes k-1}) \Pi_{I,k}\geq 0 \\
%        & \ H \in \Herm(\H).
%    \end{split}\end{align}
%%If the solution is negative for some value of $k$, then $\rho$ is $\X$-tangled. (Note that this is indeed a hierarchy-- if the minimum is negative for $k$ then it is negative for all $\tilde{k}>k$.)
%
%Computing the dual, we obtain the following characterization of $\X$-arable states.
\end{proof}
%\begin{theorem}\label{thm:xtension}
%A quantum state $\rho \in \density{\complex^n}$ is $\X$-arable if and only if for all $k \geq d$ there exists a quantum state $\sigma \in \density{(\complex^n)^{\otimes k}}$ such that the following two properties hold:
%\begin{enumerate}
%\item $\setft{Im}(\sigma) \subseteq I_k^{\perp}$, and
%\item $\tr_{k-1}(\sigma)=\rho$.
%\end{enumerate}
%(The partial trace can be taken with respect to any $k-1$ of the $k$ copies of $\complex^n$, since $\sigma$ is symmetric.)
%\end{theorem}

\begin{remark}\label{rmk:ppt}
Theorem~\ref{thm:xtension} can be strengthened at each level by requiring that the $(k,\X)$-tension $\sigma_k$ has positive partial transpose (PPT) with respect to some bipartitions of $\H^{\otimes k}$ (see \cite{Per96,horodecki2001separability} or~\cite{Wat18}), which also results in a complete hierarchy for $\X$-arability. This remark applies to many of the other results in this work.
\end{remark}

\subsection{Computing $I_k^{\perp}$}\label{sec:Ik}

In order to run the $\X$-tensions hierarchy described in Theorem~\ref{thm:xtension}, one needs to find a homogeneous ideal $I$ that cuts out out $\X$, and to describe $I_k^{\perp}$. Abstractly, one can take $I=I(\X)$, in which case
\ba\label{eq:Iperp}
I_k^{\perp} = \spn\{ \psi^{\otimes k} : \varppsi \in \X\}
\ea
(we referred to this space as $\X^k$ in the introduction). However, it can be difficult to write down a basis for this space. The conceptually cleanest way to do this is to decompose $I_k^{\perp}$ into irreducible representations of the local unitary group (see Section~\ref{sec:rep}). More concretely, it is often easier to write down a finite set of homogeneous polynomials $f_1,\dots, f_p$ that cut out $\X$, and let $I=\langle f_1,\dots, f_p \rangle.$

In this section we review some background on representation theory, and then give explicit descriptions of $I_k^{\perp}$ for all of the varieties in Example~\ref{ex:varieties}.

% Assume for simplicity that each $f_i$ is homogeneous of the same degree $d$ (which is the case for all of our examples). Viewing these as symmetric tensors $f_i \in S^d(\H^*)$, we have $I_d=\spn\{f_1,\dots, f_p\}$ and
%\ba\label{eq:Ik}
%I_k^{\perp}=(I_d \cdot S^{k-d}(\H^*))^{\perp}.
%\ea
%For those uncomfortable with the symmetric algebra, the following is an alternative: Take an orthonormal basis $v_1^*,\dots, v_{\ell}^*$ for $I_d$, and let $\Phi_{I,d}=\sum_{i=1}^{\ell} v_i v_i^* \in \Pos(\H^{\otimes d})$ be the orthogonal projection onto $I_d^* \subseteq \H^{\otimes d}$. Then
%\ba
%I_k^{\perp}=\ker(\Phi_{I,d} \otimes \I_{\H}^{\otimes k-d}) \cap S^k(\H).
%\ea

%Alternatively,
%\ba
%I_k^{\perp}= \ker((\Pi_{I_d^*} \otimes \I_{\H}^{\otimes k-d})\Pi_k),
%\ea
%where $I_d^*=\spn\{f_1^*,\dots, f_p^*\} \subseteq S^d(\H)$ is the adjoint (i.e. conjugate transpose).

%
%Then $p$ can be taken at most $\binom{n+d-1}{d}$, and $I_k$ is spanned by the degree $k-d$ monomials multiplied by the $f_i$. This gives a set of $p \binom{n+k-d-1}{k-d}$ polynomials that span $I_k$, from which one can compute a basis for $I_k^{\perp}$. Note that this procedure takes polynomial time in $n$ when $d$ and $k$ are fixed.

\subsubsection{Local unitary symmetry}\label{sec:rep}

Many varieties $\X\subseteq \P(\H)$ that are of interest in quantum information are invariant under local unitaries ${U_1 \otimes \dots \otimes U_m}$. In this case, it is clear from the expression~\eqref{eq:Iperp} that $I(\X)_k^{\perp} \subseteq S^k(\H)$ is invariant under powers of local unitaries $(U_1 \otimes \dots \otimes U_m)^{\otimes k}$, and the same is often true of other ideals $I$ that we may choose to cut out $\X$. When this is the case, we can describe $I_k^{\perp}$ by decomposing it into irreducible representations (irreps) of the local unitary group. For a finite dimensional Hilbert space $\J$, by \textit{Schur-Weyl duality} the irreps of the unitary group $\Unitary(\J)$ that appear in $\J^{\otimes d}$ are indexed by \textit{integer partitions} $\lambda \vdash d$, which are tuples of non-increasing, positive integers summing to $d$. See e.g.~\cite{landsberg2012tensors,fulton2013representation} for more details. The irrep indexed by $\lambda$ is denoted $S^{\lambda}(\J)\subseteq \J^{\otimes d}$. It follows that every irrep of the local unitary group on $S^d(\H)$ is of the form $S^{\lambda_1}(\H_1)\otimes \dots \otimes S^{\lambda_m}(\H_m)$ for some integer partitions $\lambda_i \vdash d$. %In general, one needs more than the integer partitions $\lambda_i$ to specify an irrep of $S^d(\H)$, because a given irrep may appear more than once. This is not the case in the examples we consider.

\subsubsection{Examples}\label{sec:xarable_examples}

Here we give explicit expressions for $I_k^{\perp}$ for all of the varieties in Example~\ref{ex:varieties}. These are computed using~\eqref{eq:Iperp} along with standard results that can be found e.g. in~\cite{landsberg2012tensors}.

\begin{example}[Separability]\label{ex:sep}
Let $\H=\H_1\otimes \dots \otimes \H_m$. The set of pure product states
\ba
\X_{\Sep}=\{\psi_1^{}\psi_1^* \otimes \dots \otimes \psi_m^{} \psi_m^* : \varp{\psi}{i} \in \P(\H_i)\} \subseteq \P(\H)
\ea
satisfies
\ba
I(\X_{\Sep})_k^{\perp}=S^k(\H_1)\otimes \dots \otimes S^k(\H_m).
\ea
%Hence, a state $\rho \in \density{\H}$ is separable (i.e. $\X_{\Sep}$-arable) if and only if for every positive integer $k$ there exists $\sigma \in \density{\H^{\otimes k}}$ for which
%\ba
%\setft{Im}(\sigma) \subseteq S^k(\H_1)\otimes \dots \otimes S^k(\H_m).
%\ea
%and $\tr_{k-1}(\sigma)=\rho$.
\end{example}

\begin{example}[Bosonic separability]
Let $\H=S^m(\J)$. The set of bosonic pure product states
\ba
\X_{\Sep}^{\vee}= \{(\varppsi)^{\otimes m} : \varppsi \in \P(\J)\} \subseteq \P(\H)
\ea
satisfies
\ba
I(\X_{\Sep}^{\vee})_k^{\perp}=S^{km} (\J).
\ea
%Hence, a bosonic state $\rho \in \density{\H}$ is bosonic-separable (i.e. $\X_{\Sep}^{\vee}$-arable) if and only if for every positive integer $k$ there exists $\sigma \in \density{\H^{\otimes k}}$ for which
%\ba
%\im(\sigma) \subseteq S^{m+k} (\H)
%\ea
%and $\tr_{k-1}(\sigma)=\rho$.
\end{example}

\begin{example}[Schmidt rank]\label{ex:schmidt}
Let $\H=\H_1 \otimes \H_2$. The set $\X_r \subseteq \P(\H)$ of pure states of Schmidt rank at most $r$ satisfies
\ba\label{eq:schmidt_rep}
I(\X_r)_k^{\perp} = \bigoplus_{\substack{\lambda \vdash k \\ \ell(\lambda) \leq r}} S^\lambda(\H_1)\otimes S^\lambda(\H_2).
\ea
Alternatively, $I(\X_r)$ is generated in degree $r+1$ by the span of the $(r+1)\times (r+1)$ minors, given by
\ba
I(\X_r)_{r+1}=\Lambda^{r+1} (\H_1^*)\otimes \Lambda^{r+1}(\H_2^*) \in S^{r+1}(\H^*),
\ea
from which one can compute $I_k$ directly. In concrete terms, if we let $\Pi_{\J,k}^{\w}$ be the orthogonal projection onto the antisymmetric space $\Lambda^k(\J)$, then
\ba
I(\X_r)_k^{\perp}=\ker(\Pi_{\H_1,k}^{\w}\otimes \Pi_{\H_2,k}^{\w}\otimes \I_{\H}^{\otimes k-r-1})\cap S^k(\H).
\ea
We note that $\Pi_{\H_1,k}^{\w}\otimes \Pi_{\H_2,k}^{\w}\otimes \I_{\H}^{\otimes k-r-1}$ was given the name $\Phi_r^k$ in the prior work~\cite{JLV22a} (see also Section~\ref{sec:equiv}).
% Hence, a state $\rho \in \density{\H}$ has Schmidt number at most $r$ (i.e. is $\X_r$-arable) if and only if for every positive integer k there exists $\sigma \in \density{\H^{\otimes k}}$ for which
%\ba
%\setft{Im}(\sigma) \subseteq \bigoplus_{\substack{\lambda \vdash k \\ \ell(\lambda) \leq r}} S^\lambda(\H_1)\otimes S^\lambda(\H_2)
%\ea
%and $\tr_{k-1}(\sigma)=\rho$.
\end{example}

\begin{example}[Matrix product states]\label{ex:mps}
Let $\H=\H_1\ootimes \H_m$. Let $\X_{\MPS, r}\subseteq \P(\H)$ be the set of matrix product states of bond dimension at most $r$:
\ba
\X_{\MPS, r}=\bigcap_{j=1}^{m-1} \{\varppsi \text{ of Schmidt rank $\leq r$ in the bipartition } (\textbigotimes_{i=1}^j \H_j) \otimes (\textbigotimes_{i=j+1}^m \H_j)\}.
\ea
It follows that $\X_{\MPS, r}$ is cut out in degree $r+1$ by the $(r+1)\times (r+1)$ minors of $\psi$ with respect to each of the $m-1$ bipartite cuts. These generate an ideal $I$ which can be described abstractly by
\ba
I_k^{\perp} = \bigcap_{j=1}^{m-1} \bigoplus_{\substack{\lambda \vdash k \\ \ell(\lambda) \leq r}} S^\lambda(\H_1 \otimes \dots \otimes \H_j)\otimes S^\lambda(\H_{j+1} \otimes \dots \otimes \H_m),
\ea
or more concretely,
\ba
I_k^{\perp}=\bigcap_{j=1}^{m-1} \ker(\Pi_{\H_1\ootimes \H_j, k}^{\w} \otimes \Pi_{\H_{j+1}\ootimes \H_m,k}^{\w} \otimes \I_{\H}^{\otimes k-r-1}) \cap S^k(\H).
\ea
Note that one can use this formula to compute an explicit basis for $I_k^{\perp}$ in time polynomial in $N=\dim(\H)$ if $r$, $m$, and $k$ are fixed. A similar, explicit computation of $I_k^{\perp}$ can be carried out for arbitrary \textit{tree tensor network states}, which generalize matrix product states and are characterized by having Schmidt rank $\leq r$ with respect to a collection of bipartitions forming the edges of a tree graph~\cite{barthel2022closedness}. We omit this computation for brevity.
%A state $\rho \in \density{\complex^{n_1}\otimes \dots \otimes \complex^{n_m}}$ is $\X_{MPS, r}$-arable if and only if for every positive integer $k \geq r$ there exists $\sigma \in \density{(\complex^{n_1} \otimes\dots \otimes \complex^{n_m})^{\otimes k}}$ for which
%\ba
%\setft{Im}(\sigma) \subseteq \bigcap_{j=1}^{m-1} \bigoplus_{\substack{\lambda \vdash k \\ \ell(\lambda) \leq r}} S^\lambda(\complex^{n_1}\otimes \dots \otimes \complex^{n_j})\otimes S^\lambda(\complex^{n_{j+1}}\otimes \dots \otimes \complex^{n_m})
%\ea
%and $\tr_{k-1}(\sigma)=\rho$.
\end{example}

\begin{example}[Fermionic separability]\label{ex:fermion}
Let $\H=\Lambda^m(\J)$. Let $\X_{\Sep}^{\w}\subseteq \P(\H)$ be the set of fermionic pure product states introduced in~\cite[Definition 6.1]{grabowski2011entanglement}:
\ba
\X_{\Sep}^{\w} = \{(\psi_1 \w \cdots \w \psi_m)(\psi_1 \w \cdots \w \psi_m)^* : \psi_i \in \J, \norm{\psi_i}=1\} \subseteq \P(\H)\}.
\ea
Then
\ba
I(\X_{\Sep}^{\w})_k^{\perp}=S^{k^{(m)}}(\J),
\ea
where $k^{(m)}=(\underbrace{k,\dots, k}_{m})$ and $S^{k^{(m)}}(\J)$ is the irreducible representation of $\Unitary(\J)$ indexed by $k^{(m)}$ (see Section~\ref{sec:rep}). The set $\X_{\Sep}^{\w}$ is also called the \textit{Grassmannian}. Alternatively, $\X_{\Sep}^{\w}$ is cut out by the \textit{Plücker relations}, a set of degree-2 polynomials described explicitly in~\cite[Eq. 3.4.10]{jacobson2009finite}. These generate an ideal $I$ for which one can compute $I_k^{\perp}$ explicitly by a similar method as in Example~\ref{ex:schmidt}. We omit these details for brevity.
%Hence, a fermionic state $\rho \in \density{\H}$ is fermionic separable (i.e., $\X_{\Sep}^{\w}$-arable) if and only if for every positive integer $k$ there exists $\sigma \in \density{\H^{\otimes k}}$ for which
%\ba
%\im(\sigma) \subseteq S^{k^{(m)}}(\J).
%\ea
% Explicitly,
%\ba
%S^{k^{(m)}}(\complex^n)=(\Pi_{1} \otimes \dots \otimes \Pi_m)  (\Lambda^m(\complex^n)^{\otimes k})
%\ea
\end{example}

\begin{example}[Biseparability]
Let $\H=\H_1\ootimes  \H_m$. The set of pure biseparable states
\ba
\X_{\Bisep}=\bigcup_{\emptyset \neq S \subsetneq [m]} \{ \varppsi \otimes \varpphi : \varppsi \in \P(\textbigotimes_{j \in S} \H_j), \;\;\varpphi \in \P(\textbigotimes_{j \in [m]\setminus S} \H_j) \}\subseteq \P(\H)
\ea
satisfies
\ba
I(\X_{\Bisep})_k^{\perp}=\sum_{\emptyset \neq S \subsetneq [m]} S^k(\textbigotimes_{j \in S} \H_j)\otimes S^k(\textbigotimes_{j \in [m]\setminus S} \H_j).
\ea
\end{example}

\begin{example}[Entanglement depth]
Let $\H=\H_1\ootimes  \H_m$. The set of pure $\ell$-separable states
\ba
\X_{\ell\text{-}\Sep}=\bigcup_{B \text{ an $\ell$-partition of $[m]$}} \{\varpp{\psi}{1} \ootimes \varpp{\psi}{\ell} : \varpp{\psi}{i} \in \P(\textbigotimes_{j \in B_i} \H_j)\} \subseteq \P(\H)
\ea
satisfies
\ba
I(\X_{\ell\text{-}\Sep})_k^{\perp}=\sum_{B \text{ an $\ell$-partition of $[m]$}} S^k(\textbigotimes_{j \in B_1} \H_j)\ootimes S^k(\textbigotimes_{j \in B_\ell} H_j).
\ea
The \textit{entanglement depth} of a state is defined as 1+ the largest $\ell$ for which that state is $\X_{\ell\text{-}\Sep}$-arable~\cite{guhne2005multipartite,HHH09,lucke2014detecting}.
\end{example}

\begin{example}[$t$-producible states]
Let $\H=\H_1\ootimes  \H_m$. The set of pure $t$-producible states
\ba
\X_{t\text{-prod}}= \bigcup_{\substack{B\text{ a partition of $[m]$}\\ |B_i| \geq t \text{ for all $i$}}} \{\varpp{\psi}{1} \ootimes \varpp{\psi}{\ell} : \varpp{\psi}{i} \in \P(\textbigotimes_{j \in B_i} \H_j)\} \subseteq \P(\H)
\ea
satisfies
\ba
I(\X_{t\text{-prod}})_k^{\perp} = \sum_{\substack{B\text{ a partition of $[m]$}\\ |B_i| \geq t \text{ for all $i$}}} S^k(\textbigotimes_{j \in B_1} \H_j)\ootimes S^k(\textbigotimes_{j \in B_\ell} \H_j).
\ea
\end{example}

%A similar, explicit computation of $I_k^{\perp}$ can be easily carried out for the set of pure \textit{$\ell$-separable states}: pure states that are $\ell$-wise separable across some $\ell$-element multipartition.
%
%
% as well as the set of pure \textit{$t$-producible states}: pure states that are separable across some multipartition of subsystems, each of size at least $t$. We omit this for brevity.
%%Hence, a state $\rho \in \density{\H}$ is biseparable (i.e. $\X_{\Bisep}$-arable) if and only if for every positive integer $k$ there exists $\sigma \in \density{\H^{\otimes k}}$ for which
%%\ba
%%\setft{Im}(\sigma) \subseteq \sum_{J \subseteq [m]} S^k(\otimes_{j \in J} \H_j)\otimes S^k(\otimes_{j \in [m]\setminus J} \H_j).
%%\ea
%%and $\tr_{k-1}(\sigma)=\rho$.
%\end{example}

%--------------------------------------------------------------------------------------------%
\subsection{A de Finetti theorem for fermionic separability}\label{sec:qdf}
%--------------------------------------------------------------------------------------------%

In some cases, a quantitative version of Theorem~\ref{thm:xtension}, which upper bounds the distance between a $(k,\X)$-tendable and the set of $\X$-arable states, can be obtained. This is closely related to various results which have been dubbed \textit{quantum de Finetti theorems}. Thus, we call such a quantitative result a \textit{de Finetti theorem for $\X$-arability}. In this section, we prove a de Finetti theorem for fermionic separability (i.e. $\X_{\Sep}^{\w}$-arability), and for completeness record similar results for separability ($\X_{\Sep}$-arability) and bosonic separability ($\X_{\Sep}^{\vee}$-arability), which are well-known. These examples are distinguished among the others in that $I_k^{\perp}$ is an irreducible representation of the underlying local unitary group, allowing us to apply~\cite[Theorem III.3, Remark III.4]{koenig2009most}.

We note that the following result is quite different from the ``fermionic de Finetti theorem" of~\cite{krumnow2017fermionic}, which considers a multi-mode notion of fermionic separability.
\begin{theorem}[De Finetti theorem for fermionic separability]\label{thm:fermionicdefinetti}
Let $\J$ be a Hilbert space of dimension $n$, and let $\H=\Lambda^m(\J)$ for $m\leq n$. If $\rho \in \density{\H}$ is $(k,\X_{\Sep}^{\w})$-tendable, then there exists an $\X_{\Sep}^{\w}$-arable state $\tau \in \density{\H}$ for which
\ba
\norm{\rho - \tau}_1 \leq \frac{4 m (n-m)}{n+k}.
\ea
\end{theorem}
Here, $\norm{\cdot}_1$ denotes the trace norm. The theorem assumes $m \leq n$ because otherwise $\H=0$.
\begin{proof}
We prove the theorem as a corollary to~\cite[Theorem III.3, Remark III.4]{koenig2009most}. By convexity, it suffices to assume that $\rho$ is $(k,\X_{\Sep}^{\w})$-tendable to a pure state. Let
\ba
\A&=\H\\
\B&=S^{(k-1)^{(m)}}(\J)\\
\C&=S^{k^{(m)}}(\J)\\
X&=\spn\{e_1 \w \cdots \w e_m\} \subseteq \A.
\ea
Then $\C \subseteq \A \otimes \B$ is an irreducible subrepresentation (with multiplicity one) of unitary group $\Unitary(\J)$. Applying the standard dimension formula, we obtain~\cite{itzykson1966unitary}
\ba
\dim(S^{k^{(m)}}(\J))=\prod_{\substack{i \in [m]\\j\in [{k}]}} \frac{n+j-i}{m-i+{k}-j+1}.
\ea
Let $\Pi_X$ be the orthogonal projection onto $X$ and let $\Pi_\C$ be the orthogonal projection onto $\C$. Note that $\Pi_X=\psi\psi^*$ is a pure state, $\psi^{\otimes k-1} \in \B$, and $\tr(\Pi_{\C}(\Pi_X \otimes (\psi\psi^*)^{\otimes k-1}))=1$ because $\psi^{\otimes k} \in \C$. Hence, the quantity $\delta(X)$ given in~\cite[Definition III.2]{koenig2009most} is equal to
%\ba
%\frac{\dim(\B)}{\dim(\C)}&=\prod_{i\in [m]} \frac{m-i+1}{n+k-i}\\
%&\geq \left(\frac{m+k}{n+k}\right)^m\\
%&=\left(1-\frac{n-m}{n+k}\right)^m\\
%&\geq 1-\frac{m(n-m)}{n+k},
%\ea
\ba
\frac{\dim(\B)}{\dim(\C)}&=\prod_{i\in [m]} \frac{m+k-i}{n+k-i}\\
&\geq \left(\frac{m+k}{n+k}\right)^m\\
&=\left(1-\frac{n-m}{n+k}\right)^m\\
&\geq 1-\frac{m(n-m)}{n+k}.
\ea
The desired bound then follows directly from the bound given in~\cite[Theorem III.3, Remark III.4]{koenig2009most} in terms of $\delta(X)$.
\end{proof}

For completeness, we state the quantum de Finetti theorems for separability and bosonic separability that one can obtain from~\cite[Theorem III.3, Remark III.4]{koenig2009most}. These are well known.

\begin{theorem}[De Finetti theorem for separability]\label{thm:definettigen}
Let $\H=\H_1\otimes \dots \otimes \H_m$ with $\dim(\H_i)=n_i$. If $\rho \in \density{\H}$ is $(k,\X_{\Sep})$-tendable, then there exists a separable (i.e. $\X_{\Sep}$-arable) state $\tau \in \density{\H}$ for which
\ba
\norm{\rho - \tau}_1 \leq \frac{4 m (\max_j n_j -1)}{k+1}.
\ea
\end{theorem}
%\begin{proof}[Proof of Theorem~\ref{thm:definettigen}]
%We will prove the theorem as a corollary to~\cite[Theorem III.3, Remark III.4]{koenig2009most}. By convexity, it suffices to assume that $\rho$ is $(k,\X_{\Sep})$-tendable to a pure state. Let
%\ba
%\A&=\H\\
%\B&=S^{k-1}(\H_1) \otimes \dots \otimes S^{k-1}(\H_m)\\
%\C&=S^{k}(\H_1) \otimes \dots \otimes S^{k}(\H_m)\\
%X&={e_1}^{\otimes m}.
%\ea
%Then $\C \subseteq \A \otimes \B$ is an irreducible subrepresentation (with multiplicity one) of the product unitary group $U(\H_1) \times \dots \times U(\H_m)$. It is also straightforward to check that the quantity $\delta(X)$ defined in~\cite[Definition III.2]{koenig2009most} is equal to
%\ba
%\frac{\dim(\B)}{\dim(\C)}&=\prod_{j=1}^m\frac{\binom{n_j+k-2}{k-1}}{\binom{n_j+k-1}{k}}\\
%&\geq \prod_{j=1}^m \left(1-\frac{(n_j-1)}{k+1}\right)\\
%&\geq 1-\frac{m(\max_jn_j-1)}{k+1},
%\ea
%where the second line is a standard inequality that can be found e.g. in~\cite[Eq. (7.196)]{Wat18}.
%The desired bound then follows directly from the bound given in~\cite[Theorem III.3, Remark III.4]{koenig2009most} in terms of $\delta(X)$.
%\end{proof}

\begin{theorem}[De Finetti theorem for bosonic separability]\label{thm:bosonicdefinetti}
Let $\H=S^m(\J)$ with $\dim(\J)=n$. If $\rho \in \density{\H}$ is $(k,\X_{\Sep}^{\vee})$-tendable, then there exists an $\X_{\Sep}^{\vee}$-arable state $\tau \in \density{\H}$ for which
\ba
\norm{\rho - \tau}_1 \leq \frac{4 m (n -1)}{k+1}.
\ea
\end{theorem}

\section{Optimizing a Hermitian operator over $\X$}

Let $H \in \Herm(\H)$ be a Hermitian operator, let $I \subseteq S^{\bullet}(\H^*)$ be a homogeneous ideal and let $\X=V(I)\subseteq \P(\H)$ be the variety cut out by $I$. In this section we consider optimization problems of the form
\ba\label{eq:opt}
H_{\X}^{\downarrow}:=\min_{\psi \psi^* \in \X} \ip{\psi}{H \psi} \quad\quad \text{or}\quad\quad H_{\X}^{\uparrow}:=\max_{\psi \psi^* \in \X} \ip{\psi}{H \psi}.
\ea
%\ba
%H_{\X}^{\uparrow}:=\max_{\psi \psi^* \in \X} \ip{\psi}{H \psi}
%\ea
%and
%\ba
%H_{\X}^{\downarrow}:=\min_{\psi \psi^* \in \X} \ip{\psi}{H \psi}.
%\ea
We obtain the following hierarchy of eigencomputations for computing $H_{\X}^{\uparrow}$. Recall that we define $\Pi_{I,k}\in \Pos(\H^{\otimes k})$ to be the orthogonal projection onto $I_k^{\perp}$. Recall also that in Section~\ref{sec:Ik} we have computed $I_k^{\perp}$ explicitly for all of the varieties in Example~\ref{ex:varieties}.

%Given a collection of homogeneous polynomials $f=(f_1,\dots, f_p)$ and a hermitian matrix $W$, our main result is a hierarchy of maximum eigenvalue computations for computing $h_{f}(W)$. Let $I=\langle f_1,\dots, f_p \rangle \subseteq \complex[x]$ be the ideal generated by $f_1,\dots, f_p$. We will view $I$ as an ideal in the symmetric algebra $S(\complex^n)$ under the isomorphism $S(\complex^n) \cong \complex[x]$. For each positive integer $k$, let $\Pi_{f,k} \in \End(S^k(\complex^n))$ be the orthogonal projection onto $I_k^{\perp} \subseteq S^k(\complex^n)$, where $I_k^{\perp}$ denotes the orthogonal complement with respect to the \textit{bilinear pairing} $\ip{v}{w}=v^\t w$. We will often view $\Pi_{f,k}$ as an element of $\End((\complex^n)^{\otimes k})$ by setting it equal to zero on the orthogonal complement to $S^k(\complex^n)$.

\begin{theorem}\label{thm:robust_hier}
For each positive integer $k$, let
\ba
H_k:=\Pi_{I,k}(H \otimes \I_{\H}^{\otimes k-1}) \Pi_{I,k}
\ea
and $\nu_k=\lambda_{\max}(H_k)$. Then $\nu_1,\nu_2, \nu_3,\dots$ is a non-increasing sequence with
\ba
H_{\X}^{\uparrow}=\lim_{k \rightarrow \infty} \nu_k.
\ea
\end{theorem}

In Section~\ref{sec:poly} we relate this to the \textit{Hermitian sum of squares} (HSOS) hierarchy developed in~\cite{d2009polynomial}, proving the (non-trivial) fact that the HSOS hierarchy is equivalent to a spectral hierarchy in some settings. This result can also be proven using~\cite[Theorem 1]{catlin1999isometric} if one is familiar with vector bundles. We prove it as a consequence of Theorem~\ref{thm:xtension}.

\begin{proof}[Proof of Theorem~\ref{thm:robust_hier}]
%Note that if $\psi \psi^* \in \X$, then $\psi^{\otimes k} \in I_k^{\perp}$ for all $k$. Hence,
%    \begin{align*}
%        \nu_k & \geq \ip{\psi^{\otimes k}}{\Pi_{I,k} (H \otimes \I_{\H}^{\otimes k-1}) \Pi_{I,k} \psi^{\otimes k}} \\
%       & = \ip{\psi^{\otimes k}}{(H \otimes \I_{\H}^{\otimes k-1}) \psi^{\otimes k}} \\
%        & = \ip{\psi}{ H \psi}.
%    \end{align*}
%
%
%$\ket{v}^{\otimes k}\in I_k^{\perp}=\im(\Pi_{f,k})$. It follows that
%    \begin{align*}
%        \nu_k & \geq \bra{v}^{\otimes k}\left(\Pi_{f,k} \Big(W \otimes \I_{n}^{\otimes k-1}\Big) \Pi_{f,k}\right) \ket{v}^{\otimes k} \\
%       & = \bra{v}^{\otimes k} \Big(W \otimes \I_{n}^{\otimes k-1}\Big)  \ket{v}^{\otimes k} \\
%        & = \bra{v} W \ket{v}.
%    \end{align*}
%    Maximizing over all $\psi \psi^* \in \X$ then shows that $H_{\X}^{\uparrow}$ is bounded above by $\nu_k$.
    
%   To see that the sequence of $\nu_k$'s is non-increasing, let $v \in I_k^{\perp}$ be an eigenvector of $H_k$ with maximum eigenvalue and unit length. Then $\H^* \hook v \in I_{k-1}^{\perp}$, so $\tr_{1}(vv^*) \in \conv\{uu^* : u \in I_{k-1}^{\perp}\}$, and hence $\nu_k = \ip{v}{H_k v} = \tr(H_{k-1} \tr_1(vv^*)) \leq \nu_{k-1}$. Since the sequence of $\nu_k$'s is also bounded below (by $\lambda_{\textup{min}}(H)$, for example), we know that $\nu := \lim_{k\rightarrow \infty} \nu_{k}$ exists. All that remains is to show that it equals $W_{\X}^{\uparrow}$.
    
Observe that $\nu_k$ is the optimal value of the following semidefinite program:
    \begin{align}\begin{split}
        \textup{maximize:} & \ \tr\big((H \otimes I_{\H}^{\otimes (k-1)}) \sigma_k\big) \\
        \textup{subject to:} & \ \im(\sigma_k) \subseteq I_k^{\perp} \\
        & \ \sigma_k \in \Density(\H^{\otimes k}),
    \end{split}\end{align}
which can be reformulated as
    \begin{align}\begin{split}\label{eq:sdp_robust_proof}
        \textup{maximize:} & \ \tr (H \rho) \\
        \textup{subject to:} & \ \rho \text{ is } (k,\X)\text{-tendable}.
    \end{split}\end{align}

This immediately shows that $H_{\X}^{\uparrow}$ is bounded above by $\nu_k$, because every $\X$-arable state is $(k,\X)$-tendable. It furthermore shows that the $\nu_k$ are non-increasing, because the partial trace of a $(k,\X)$-tension is a $(k-1,\X)$-tension. Since the $\nu_k$ are also bounded below (by the minimum eigenvalue of $H$, for example), $\lim_k \nu_k$ exists. Let $\rho_k \in \density{\H}$ be a state for which the optimum value of the SDP~\eqref{eq:sdp_robust_proof} is attained. By compactness of $\density{\H}$, this sequence contains a subsequence converging to some $\rho\in \density{\H}$, so $\lim_k \nu_k = \tr(H \rho)$. Clearly $H_{\X}^{\uparrow}\leq \tr(H \rho)$, since this inequality holds for each $\rho_k$. Conversely, since $\rho$ is $(k,\X)$-tendable for arbitrary $k$, $\rho$ is $\X$-arable by Theorem~\ref{thm:xtension}, so $H_{\X}^{\uparrow}\geq \tr(H \rho)$. This completes the proof.
\end{proof}

%\begin{remark}
%Note that computing the matrix $\Pi_{I,k}$ could be impractical -- In particular, for each $k$ one needs to compute an orthonormal basis for $I_k^{\perp}$. It therefore may be convenient to replace $\Pi_{I,k}$ by a different matrix $\tilde{\Pi}_{I,k}$ that is easier to compute at each level but still gives rise to a convergent hierarchy. For example, Theorem~\ref{thm:robust_hier} also holds (by a trivial modification of the proof) when the matrices $\Pi_{I,k}$ are replaced by any contractions $\norm{\tilde{\Pi}_{I,k}}_{\infty}\leq 1$ with $+1$ eigenspace equal to $I_k^{\perp}$. For example, if $I$ is generated in degree at most $d$, then one can take $\tilde{\Pi}_{I,k}=\Pi_k (\Pi_{I,d} \otimes \I_{\H}^{\otimes k-d})$ for all $k \geq d$ (recall $\Pi_k$ is the orthogonal projection onto $S^k(\H)$).
%\end{remark}

%More generally, one can obtain similar de Finetti theorems for $\X$-arability for any conic variety $\X$ living in a tensor product space for which $I(\X)_k^{\perp}$ is a \textit{Cartan power}, meaning that it is an irreducible representation of the underlying unitary group. We will not dwell on this point, as we are unaware of further examples that might be of interest in quantum information theory.

\subsection{$\X$-tanglement witnesses}

We say that $H \in \Herm(\H)$ is \textit{$\X$-positive} if $H_{\X}^{\downarrow} \geq 0$. We say that $H$ is an \textit{$\X$-tanglement witness} if it is $\X$-positive and has at least one negative eigenvalue. The set of $\X$-arable states forms a closed, convex cone, so any $\X$-tangled state admits a separating hyperplane from the set of $\X$-arable states:

\begin{prop}\label{prop:xpositive}
A state $\rho \in \density{\H}$ is $\X$-arable if and only if $\tr(\rho H) \geq 0$ for every $\X$-tanglement witness $H \in \Herm(\H)$.
\end{prop}

Theorem~\ref{thm:robust_hier} implies the following hierarchy for certifying that a given operator $H$ is an $\X$-tanglement witness.

\begin{prop}
$H \in \Herm(\H)$ is an $\X$-tanglement witness if and only if $H \notin \Pos(\H)$ and $H_k:=\Pi_{I,k} (H\otimes \I_{\H}^{\otimes k-1})\Pi_{I,k}$ is positive semidefinite for $k \gg 0$.
\end{prop}

%
%shows that $H_{\X}^{\downarrow}$ is lower bounded by the minimum eigenvalue of $H_k=\Pi_{I,k} (H\otimes \I_{\H}^{\otimes k-1})\Pi_{I,k}$, and thus gives a method of certifying that a given $H$ is an $\X$-tanglement witness.

\section{$\X$-tangled subspaces}\label{sec:xtangled_sub}
Let $I \subseteq S^{\bullet}(\H^*)$ be a homogeneous ideal and let $\X=V(I)\subseteq \P(\H)$ be the corresponding projective variety. For a subspace $\U \subseteq \H$, we say that $\U$ is \textit{$\X$-tangled} if
\ba
\{\psi\psi^* : \psi \in \U \} \cap \X=\emptyset.
\ea
We define the \textit{geometric measure of $\X$-tanglement} (GM$\X$) of $\U$ to be
\ba
E_\X(\U)=1-(\Pi_{\U})_{\X}^{\uparrow},
\ea
where $\Pi_{\U}$ denotes the orthogonal projection onto $\U$. Note that $E_{\X}(\U)>0$ if and only if $\U$ is $\X$-tangled. When $\X=\X_{\Sep}$, this specializes to the~\textit{geometric measure of entanglement} (GME) of subspaces studied in~\cite{demianowicz2019entanglement,zhu2024quantifying}, and when $\U=\spn\{\psi\}$ is one-dimensional this becomes the well-studied GME of $\psi$~\cite{wei2003geometric}.

%\ba
%E_\X(\U):= \frac{1}{2} \min_{\substack{\varppsi \in \check{\U} \\ \varpphi \in \X}} \biggnorm{\varppsi-\varpphi}_1.
%%h_\X(\U):=\max_{\substack{\ket{u}\in \U\\ \ket{x} \in \X}} \abs{\braket{u}{x}}^2.
%\ea
Theorem~\ref{thm:robust_hier} gives a hierarchy of maximum eigenvalue computations to determine $E_{\X}(\U)$.
\begin{cor}\label{cor:subspace}
For each positive integer $k$, let $\nu_k=\lambda_{\max}(\Pi_{I,k} (\Pi_{\U} \otimes \I_{\H}^{\otimes k-1}) \Pi_{I,k}).$ Then $\nu_1 \geq \nu_2 \geq \dots$ and $E_\X(\U)= 1-\lim_{k\rightarrow \infty} {\nu_k}$.
\end{cor}

\subsection{$\X$-tangled states from $\X$-tangled subspaces}

The \textit{range criterion} is a well-known method of certifying entanglement of a mixed state. This generalizes straightforwardly to $\X$-tanglement, as follows:

\begin{prop}
Let $\rho \in \density{\H}$ be a state. If $\im(\rho)\subseteq \H$ is an $\X$-tangled subspace, then $\rho$ is $\X$-tangled.
\end{prop}
\begin{proof}
We prove the contrapositive. Suppose $\rho$ is $\X$-arable with decomposition $\rho=\sum_i p_i \varp{\psi}{i}$. Then $\psi_i \in \im(\rho)$ and $\varp{\psi}{i} \in \X$, so $\im(\rho)$ is not an $\X$-tangled subspace.
\end{proof}

More generally, \textit{highly} $\X$-tangled subspaces can be used to certify~\textit{highly $\X$-tangled} mixed states: States that remain $\X$-tangled after unitary perturbations. The following is a straightforward generalization of~\cite[Theorem 2]{zhu2024quantifying}.

\begin{theorem}
Let $\rho \in \density{\H}$ be a state, and let $\U=\im(\rho)\subseteq \H$. Then for any $H\in \Herm(\H)$ with $\norm{H}_1 < E_{\X}(\U)^{1/2}$, the state $\rho'=e^{iH} \rho e^{-iH}$ is $\X$-tangled.
\end{theorem}

\subsection{$\X$-tanglement witnesses from $\X$-tangled subspaces}

$\X$-tangled subspaces can be used to construct $\X$-tanglement witnesses, as follows~(see also~\cite[Section 5]{LJ21}):
\begin{prop}\label{prop:sub_witness}
Let $\U\subseteq \H$ be an $\X$-tangled subspace, and let $\mu=1/{(\Pi_{\U})_{\X}^{\uparrow}}$. Then $\mu>1$ and $H:=  \I - \mu \Pi_{\U}$ is an $\X$-tanglement witness with $\dim(\U)$ negative eigenvalues of magnitude $\mu-1$.
\end{prop}
The number and magnitude of negative eigenvalues roughly measures ``how good" an $\X$-tanglement witness $H$ is. This proposition shows that good witnesses can be constructed using subspaces of large dimension and high geometric measure.

\section{Deciding subspace $\X$-tanglement using Hilbert's Nullstellensatz}\label{sec:null}
In this section we consider the special case of simply deciding if $\U$ is $\X$-tangled, without determining the geometric measure $E_{\X}(\U)$. Hilbert's Nullstellensatz gives a hierarchy of linear systems for this problem, extending the hierarchy introduced in~\cite{JLV22a} to general $\X$-tanglement:
% and was previously studied by the last three authors in~\cite{JLV22a,JLV22b}:

\begin{theorem}[Nullstellensatz hierarchy for subspace $\X$-tanglement]\label{thm:null_sub}
A subspace $\U\subseteq \H$ is $\X$-tangled if and only if $I_k + I(\U)_k= S^k(\H^*)$ for $k \gg 0$.
\end{theorem}
\begin{proof}
$\U$ is $\X$-tangled if and only if $V(I + I(\U))=\emptyset$. By Hilbert's Nullstellensatz (Theorem~\ref{thm:hilbert}), this is equivalent to $I_k + I(\U)_k= S^k(\H^*)$ for $k \gg 0$.
\end{proof}

\begin{remark} One can check that if $I_k + I(\U)_k= S^k(\H^*)$, then $I_{k+1} + I(\U)_{k+1}= S^{k+1}(\H^*)$, justifying our use of the word ``hierarchy." \end{remark}

The $k$-th level of this hierarchy checks if $I_k + I(\U)_k= S^k(\H^*)$, which amounts to solving a linear system in the $\binom{N+k-1}{k}$-dimensional space $S^k(\H^*)$, where $N=\dim(\H)$.  Because this is a computation in the degree-$k$ part of the polynomial ring, we also call $k$ the \textit{degree} of the hierarchy. The computational complexity of using this hierarchy to certify $\X$-tanglement of $\U$ depends on how large $k$ needs to be for the equality $I_k + I(\U)_k= S^k(\H^*)$ to hold. We refer to bounds on $k$ as \textit{degree bounds} for the hierarchy.
% and refer to bounds on $k$ as \textit{degree bounds}. We can use such degree bounds  the computational complexity of using this hierarchy to certify that a subspace $\U$ is $\X$-tangled depends on how large $k$ needs to be.

Since determining $\X$-tanglement of a subspace is NP-hard, one should expect the worst-case degree to grow at least linearly in $N$ in general. 
%%If $\X$ is \textit{(arithmetically) Cohen-Macaulay}, then the worst-case degree is captured by the \textit{(Castelnuovo-Mumford) regularity} of $\X$ (see Section~\ref{sec:commutative_algebra} for the definitions). This allows us to precisely determine the worst case degree for many of the varieties in Example~\ref{ex:varieties}. Even when $\X$ is not Cohen-Macaulay, we prove that if $I$ is generated in degree at most $d$ then the worst case degree is upper bounded by $N(d-1)+1$.
We say that a property holds for a \textit{generically chosen} element if it holds on a Zariski open dense (full measure) subset of the underlying Hilbert space. We prove both worst case and generic degree bounds. Our starting point is the following classical fact (see~\cite{harris2013algebraic}).
{\begin{fact}\label{fact:dim}\hspace{1em}
Let $\X\subseteq \P(\H)$ be a variety, and let $\dim(\X)$ be the (affine) dimension of $\X$ (see Section~\ref{sec:commutative_algebra}). Then:
\begin{enumerate}
\item Every subspace $\U$ of dimension $\dim(\U) > N - \dim(\X)$ intersects $\X$.
\item A generically chosen subspace $\U$ of dimension $\dim(\U) \leq N - \dim(\X)$ is $\X$-tangled.
\end{enumerate}
\end{fact}}

So the relevant regime to consider is $\dim(\U) \leq N - \dim(\X)$, otherwise the problem is trivially solvable in polynomial time. While a generically chosen subspace satisfying this bound will be $\X$-tangled, it is not clear at what level the Nullstellensatz hierarchy will certify it to be so.

In Sections~\ref{sec:worst_case} and~\ref{sec:generic} we prove general purpose worst-case and generic degree bounds for the Nullstellensatz hierarchy, respectively. 
%We prove that for many varieties of interest and any $0 < \eps < 1$, a generically chosen subspace of dimension $\dim(\U)\leq (1-\eps)(N- \dim(\X))$ will be certified to avoid $\X$ at a constant level of the hierarchy. This gives a polynomial time algorithm in the generic setting, with an arbitrarily small multiplicative loss in dimension. This also strengthens the results of~\cite{JLV22a}, which only hold for particular choices of $\X$ and $\eps$ (for example, when $\X=\X_{1}$ is the set of bipartite pure product states, these results take $\eps=3/4$). Our improvements come from using more sophisticated algebraic techniques.
In Section~\ref{sec:examples} we apply these degree bounds to the varieties in Example~\ref{ex:varieties}. A collection of these examples is presented in Table~\ref{table:nullstellensatz}. In all cases, we find that $\dim(\X)=o(N)$, and a generically chosen subspace of dimension $(1-\eps)N$ is certified $\X$-tangled at a constant degree $k$. Furthermore, in all cases we prove that $k=O(N)$ is sufficient in the worst case. We improve this bound in some cases (see Theorem~\ref{thm:worst_poly}).

An alternative hierarchy to decide if $\U$ is $\X$-tangled can be obtained from Corollary~\ref{cor:subspace}, which implies that $\U$ is $\X$-tangled if and only if $\lambda_{\max}(\Pi_{I,k} (\Pi_{\U} \otimes \I_{\H}^{\otimes k-1})\Pi_{I,k})<1$ for $k \gg 0$. Satisfyingly, in Section~\ref{sec:equiv} we prove that this is equivalent to the Nullstellensatz hierarchy. We also clarify the relationship between the Nullstellensatz hierarchy of Theorem~\ref{thm:null_sub} and the hierarchy introduced in~\cite{JLV22a}.

%When determining if $\U$ is $\X$-tangled, one should first verify that $\dim(\U) \leq n - \dim(\X)$, for otherwise  $\U$ must intersect $\X$.
%
% We prove that for many varieties of interest, a generic subspace of dimension $\dim(\U)\leq (1-\eps)(n- \dim(\X))$ is certified to avoid $\X$ 
%
%
%We prove that for many varieties $\X$ of interest, 

%
%
% then it is well-known that no subspace $\U$ of dimension $\dim(\U) > n - \dim(\X)$ is $\X$-tangled, and that a generic subspace of dimension 
%
%Despite the fact that the worst-case degree must grow linearly in $n$, we prove that in the generic setting the degree can be chosen constant in $n$. 
%
%% $\X\in \{\X_{\Sep},\X_{\Sep}^{\vee}, \X_r, \X_{\Bisep}\}$.
%
%
%
%
% Hence, the computational complexity of using this hierarchy to certify $\X$-tanglement depends on how large $k$ needs to be. 
%
%
%
%
%
%It is useful to control how large $k$ needs to be in different scenarios, as this determines the size of the computation. In this section, we prove bounds on $k$ that apply in both the worst case and average case scenarios. Since determining $\X$-tanglement of a subspace is NP-hard in general, we can expect $k$ to grow at least linearly in $n=\dim(\H)$ in the worst case. It is well-known that the
%
%
%It is useful to have bounds on how large $k$ needs to be, as this controls the size of computation. 
%As $k$ increases, the size of the computation increases. In this section we prove guarantees on how large 

\subsection{Commutative algebra background}\label{sec:commutative_algebra}

Our degree bounds for the Nullstellensatz hierarchy require more background in commutative algebra, which we now briefly review. See~\cite{eisenbud2005geometry,eisenbud2013commutative,cox2013ideals} for more details.

Let $R=S^{\bullet}(\H^*)$, let $I \subseteq R$ be a homogeneous ideal, and let $A$ be a (commutative) $\complex$-algebra (which we often take to be $R/I$). We say that $f_1,\dots, f_{\ell} \in R$ are \textit{algebraically independent} if they do not satisfy any non-trivial polynomial equation, i.e. $g(f_1,\dots, f_{\ell})= 0 \Rightarrow g=0$. An \textit{$A$-module} $M$ is a ``vector space over $A$": An additive abelian group with scalar multiplication ${A\times M \rightarrow M}$ that distributes over addition in $A$ and in $M$, and satisfies $(ab)m=a(bm)$ for all $a,b \in A, m \in M$. A \textit{homomorphism} of modules is a map $\phi$ for which $\phi(a_1m_1+a_2m_2)=a_1\phi(m_1)+a_2\phi(m_2)$. %We will commonly take $A=\field[f_1,\dots,f_{\ell}]$ the algebra generated by the $f_i$.
We say that $M$ is a \textit{finite} $A$-module if there exist $m_1,\dots,m_p\in M$ for which every element of $M$ is an $A$-linear combination of the $m_i$.

%We say that $A$ is \textit{Noetherian} if every increasing chain of ideals $I_1 \subseteq I_2 \subseteq \dots \subseteq A$ stabilizes, meaning that $I_k=I_{k+1}=\cdots $ for $k$ large enough. By Hilbert's basis theorem, the symmetric algebra $R$ is Noetherian, as well as $R/I$ for any ideal $I\subseteq R$.

We will make use of the following simplified version of Noether's Normalization lemma:

\begin{lemma}[Noether Normalization]
Let $f_1,\dots,f_{\ell} \in R$ be homogeneous of the same degree, and let $A=\complex[f_1,\dots, f_{\ell}]$ be the algebra they generate. Then there exist algebraically independent elements $g_1,\dots, g_p \in \spn\{f_1,\dots, f_{\ell}\}$ for which $A$ is a finite $\complex[g_1,\dots, g_p]$-module.
\end{lemma}

A \textit{grading} of $A$ is a direct sum decomposition $A=\textbigoplus_{d=0}^{\infty} A_d$ into $\complex$-vector spaces $A_d$ for which $A_c \cdot A_d \subseteq A_{c+d}$. The vector space $A_d$ is called the \textit{$d$-th graded part} (or \textit{homogeneous part of degree $d$}) of $A$. The symmetric algebra is graded by degree $R_d= S^d(\H^*)$, as is the quotient $R/I=\textbigoplus_d R_d/I_d$. For a graded algebra $A$, define the \textit{Hilbert function} $\hf_A(d)=\dim(A_d)$ and the \textit{Hilbert series} $\hs_A(t)=\sum_{d=0}^{\infty} \hf_A(d) t^d$.

For a graded algebra $A$ and integer $c$, let $A(-c)$ be the graded algebra with $d$-th graded part $A_{d-c}$ if $c \leq d$ and $0$ otherwise. If $A$ is a graded algebra, then we say $M$ is a \textit{graded $A$-module} if there is a direct sum decomposition $M=\textbigoplus_{d=0}^{\infty} M_d$ of abelian groups for which $A_c M_d \subseteq M_{c+d}$. The shifted modules $M(-c)$ are defined similarly. A \textit{homomorphism of graded algebras/modules} $\phi:M\rightarrow N$ is an algebra/module homomorphism for which $\phi(M_d)\subseteq N_d$. An \textit{exact sequence} of graded algebras/modules is a sequence of graded algebra/module homomorphisms
\ba
M^{(k)}\overset{\phi_k}{\longrightarrow} M^{(k-1)} \overset{\phi_{k-1}}{\longrightarrow} M^{(k-2)} \overset{\phi_{t-2}}{\longrightarrow} \cdots
\ea
for which $\im(\phi_i)=\ker(\phi_{i-1})$ for all $i$. A \textit{(graded) minimal free resolution} of $A=R/I$ is an exact sequence of the form
\ba
0 \longrightarrow \textbigoplus_{j=1}^{\beta_t} R(-c_{t,j}) \overset{\phi_{t}}{\longrightarrow} \textbigoplus_{j=1}^{\beta_{t-1}} R(-c_{t-1,j}) \overset{\phi_{t-1}}{\longrightarrow} \dots \overset{\phi_{1}}{\longrightarrow} \textbigoplus_{j=1}^{\beta_0} R(-c_{0,j}) \overset{\phi_{0}}{\longrightarrow} R/I \overset{\phi_{-1}}{\longrightarrow} 0,
\ea
where each arrow is a homomorphism of graded $A$-modules, and each $\beta_i$ is minimal (= the minimum number of generators of $\ker(\phi_{i-1})$). The \textit{Castelnuovo-Mumford regularity} of $A$ is defined as $\reg(A):=\max_{i,j}\{c_{i,j}-j\}$. It is a non-trivial but well-known fact that a minimal free resolution exists for any homogeneous ideal $I$, and that the regularity is independent of the choice of minimal free resolution~\cite{eisenbud2005geometry}.

Let $A$ be a graded algebra. A collection of homogeneous elements $a_1,\dots, a_t \in A$ is called \textit{regular} if each $a_i$ is not a zero divisor in $A/\langle a_1,\dots, a_{i-1}\rangle$. The \textit{depth} $\setft{depth}(A)$ is the maximum length of a regular sequence of homogeneous elements of positive degree. An ideal $P \subseteq A$ is \textit{prime} if $P\neq A$ and $(p,q\in A, pq \in P \Rightarrow p \in P \text{ or } q\in P)$. The \textit{(Krull) dimension} $\dim(A)$ is the largest $d$ for which there exists a strictly ascending chain $0=P_0 \subsetneq P_1 \subsetneq P_2 \subsetneq \dots \subsetneq P_d$ of prime ideals contained in $A$. For a projective variety $\X=V(I)$, we let $\dim(\X):=\dim(R/I)$ be the \textit{affine (Krull) dimension} of $\X$.  It is known that $\setft{depth}(A)\leq \dim(A)$, and we say that $A$ is \textit{Cohen-Macaulay} if equality holds. The \textit{radical} of $I$ is defined as $\sqrt{I}=\{f : f^k \in I \text{ for some k}\}$. By Hilbert's Nullstellensatz, $I(\X)=\sqrt{I}$ for any ideal $I$ cutting out $\X$.

%We say that an $A$-module $P$ is \textit{projective} if there is an $A$-module $Q$ for which $P \oplus Q \cong \sum_j R(-c_j)$ for some $c_j$. A \textit{(graded) projective resolution} of $R/I$ is an exact sequence
%\ba
%0 \rightarrow P_t \rightarrow P_{t-1} \rightarrow \cdots \rightarrow P_1 \rightarrow R/I \rightarrow 0
%\ea
%where each $P_i$ is a projective $R/I$-module. We define the 

\subsection{Worst-case degree bounds}\label{sec:worst_case}

Let $\H$ be a complex Hilbert space of dimension $N$, let $R=S^{\bullet}(\H^*)$, let $\X \subseteq \P(\H)$ be a variety and let $I\subseteq R$ be a homogeneous ideal cutting out $\X$ and generated in degree at most $d$. In this section we prove worst-case degree bounds for the Nullstellensatz hierarchy. We remark that the results of this section hold more generally over any algebraically closed field.

In the case when $R/I$ is Cohen-Macaulay, the following classical result shows that $\reg(R/I)$ determines the worst-case degree needed for the Nullstellensatz hierarchy. See~\cite[Proposition 4.14]{eisenbud2005geometry}

%determines the precise degree at which the nullstellensatz hierarchy decides $\X$-tanglement of a subspace when $R/I$ is {Cohen-Macaulay}.
% (meaning that its depth equals its dimension-- see e.g.~\cite[Appendix A2E]{eisenbud2005geometry}).

\begin{theorem}\label{thm:worst_case}
Suppose that $R/I$ is Cohen-Macaulay, and let $\U \subseteq \H$ be a subspace. Then the following statements are equivalent:
\begin{enumerate}
\item $\U$ is $\X$-tangled.
\item $I_k + I(\U)_k = R_k\;\;$ for $k = \setft{reg}(R/I)+1$.
\end{enumerate}
Furthermore, the bound $k = \setft{reg}(R/I)+1$ is tight: There exists an $\X$-tangled subspace $\U$ such that $I_k + I(\U)_k \subsetneq R_k\;\;$ for $k = \setft{reg}(R/I)$.
\end{theorem}

The next theorem gives a worst-case degree bound in the case when $R/I$ may not be Cohen-Macaulay.

\begin{theorem}\label{thm:worst_case_bound}
For a linear subspace $\U \subseteq \H$ the following statements are equivalent:
\begin{enumerate}
\item $\U$ is $\X$-tangled.
\item $I_k + I(\U)_k = R_k\;\;$ for $k = N(d-1)+1$.
\end{enumerate}
\end{theorem}
\begin{proof}
$2 \Rightarrow 1$ is clear, so it suffices to prove $1 \Rightarrow 2$. Let $f_1,\dots, f_{\ell} \in R_d$ be a vector space basis for $I_d+I(\U)_d$. Then $\langle f_1, \dots, f_{\ell} \rangle_D = I_D + I(\U)_D$ for all $D \geq d$. Let $A=\complex[f_1,\dots, f_{\ell}]$ and $M_A$ be the ideal generated by $f_1,\dots, f_{\ell}$ in $A$. We need to prove that $\langle f_1,\dots, f_{\ell} \rangle_k=R_k$. By Noether's normalization lemma, there exist algebraically independent elements $g_1,\dots, g_{\tilde{N}} \in \spn\{f_1,\dots, f_{\ell}\}$ such that $A$ is a finite $B:=\complex[g_1,\dots, g_{\tilde{N}}]$-module. Since $g_1,\dots, g_{\tilde{N}}$ are algebraically independent, $\tilde{N} \leq N$. Let $M_B$ be the ideal generated by $g_1,\dots, g_{\tilde{N}}$ in $B$.
\begin{claim}
$V(g_1,\dots, g_{\tilde{N}})=\emptyset$.
\end{claim}
Before proving the claim, we use it to prove the theorem. Since at least $N$ equations are needed to cut out the empty set, the claim implies $\tilde{N} \geq N$, and hence $\tilde{N}=N$. For each $j=\{0,1,\dots, N\}$, let $K_j=R/(R \langle g_1,\dots, g_j \rangle)$. Then we have an exact sequence of graded algebras
\ba
0 \rightarrow K_j(-d) \xrightarrow{g_j} K_j \rightarrow K_{j+1} \rightarrow 0.
\ea
Thus,
\ba
\HS_{K_j}(t)=\HS_{K_j(-d)}(t)+\HS_{K_{j+1}}(t).
\ea
Substituting $\HS_{K_j(-d)}(t)=t^d\HS_{K_j}(t),$ we obtain
\ba
\HS_{K_{j+1}}(t)=(1-t^d) \cdot \HS_{K_j}(t).
\ea
Thus,
\ba
\HS_{K_N}(t)&=(1-t^d)^N \cdot \HS_{R}(t)\\
&=\left(\frac{1-t^d}{1-t}\right)^N\\
&=(t^{d-1} + t^{d-2}+ \dots + t+1)^N.
\ea
So the Hilbert series of $K_N$ is equal to zero in degree $k=N(d-1)+1$. Equivalently, $R \langle g_1,\dots, g_N\rangle_k=R_k.$ Since $R_k \supseteq \langle f_1,\dots, f_{\ell} \rangle \supseteq \langle g_1,\dots, g_N \rangle$, this implies $\langle f_1,\dots, f_{\ell} \rangle_k=R_k$. This completes the proof modulo proving the claim.

To prove the claim, we first observe that $\sqrt{A M_B}=M_A$. First note that $A M_B \subseteq M_A$, so $\sqrt{A M_B} \subseteq \sqrt{M_A}=M_A$. For the reverse inclusion, it suffices to prove that any homogeneous $h \in M_A$ is contained in $\sqrt{A M_B}$. By~\cite[Proposition 7.1]{atiyah1994introduction} it holds that $h^k = p_0+p_1 h+\dots + p_{k-1} h^{k-1}$ for some $k$ and $p_i \in B$. Since $h$ is homogeneous of positive degree, we can assume that each $p_i$ is homogeneous of positive degree, so $p_i \in M_B$. This shows that $h^k \in A M_B$, so $h \in \sqrt{A M_B}$.

Hence,
\ba
\sqrt{R \; M_B}=\sqrt{R \; A \; M_B}=\sqrt{R \sqrt{A \; M_B}}=\sqrt{R \; M_A}=\langle x_1,\dots, x_N\rangle,
\ea
where the first equality follows from $R=RA$, the second follows from $\sqrt{IJ}=\sqrt{\sqrt{I}\sqrt{J}}$, the third is by above, and the fourth follows from the assumption $V(f_1,\dots, f_{\ell})=0$.

%Since $V(M_A)=\U \cap \X=\{0\}$, we conclude
Thus,
\ba
V(g_1,\dots, g_{\tilde{N}})=V(R \; M_B)=V(x_1,\dots, x_N)=0.
\ea
This completes the proof.
\end{proof}

%
%We say that $\X$ is \textit{Cohen-Macaulay} if $\complex[x]/I(\X)$ is Cohen Macaulay (see e.g. \cite[Appendix A2E]{eisenbud2005geometry}). The following theorem determines the degree at which the nullstellensatz decides $\X$-tanglement of a subspace when $\X$ is Cohen-Macaulay.

%Let $R= \complex[x]$, let $I \subseteq R$ be a homogeneous ideal, and let $\X=V(I)$. For a homogeneous ideal $I \subseteq R$, we let $\setft{reg}(R/I)$ be the \textit{(Castelnuovo-Mumford) regularity} of $R/I$ (see e.g. \cite[Chapter 4A]{eisenbud2005geometry}).

Note that Theorems~\ref{thm:worst_case_bound} and~\ref{thm:worst_case} combined show that if $R/I$ is Cohen-Macaulay and $I$ is generated in degree at most $d$, then $\setft{reg}(R/I) \leq N(d-1)$.

%\begin{theorem}\label{thm:worst_case2}
%Let $R= \complex[x]$, let $I \subseteq R$ be a homogeneous ideal generated in degree $d$, and let $\X=V(I)$. Then for any linear subspace $\U \subseteq \complex^n$ the following statements are equivalent:
%\begin{enumerate}
%\item $\U \cap \X = \{0\}$.
%\item $I_k + I(\U)_k = R_k\;\;$ for $k = \setft{reg}(R/I)+1$.
%\end{enumerate}
%Furthermore, the bound $k = \setft{reg}(R/I)+1$ is tight: There exists a linear subspace $\U$ for which $\U \cap \X=\{0\}$ and $I_k + I(\U)_k \subsetneq R_k\;\;$ for $k = \setft{reg}(R/I)$.
%\end{theorem}

\subsection{Generic degree bounds}\label{sec:generic}

As before, let $\H$ be a complex Hilbert space of dimension $N$, let $R=S^{\bullet}(\H^*)$, let $\X \subseteq \P(\H)$ be a variety and let $I$ be a homogeneous ideal cutting out $\X$. In this section we prove generic degree bounds for the Nullstellensatz hierarchy.

A \textit{Zariski open set} is the set complement to a projective variety. We say that a property holds for a \textit{generically chosen} $s$-dimensional subspace $\U \subseteq \H$ if it holds on a Zariski open dense subset of the set of all $s$-dimensional subspaces (which forms a variety called the \textit{Grassmannian}). In particular, if a basis for $\U$ is chosen at random with respect to the Haar measure, then the property holds with probability one.

The following results hold over any infinite field. The first result shows that a generic version of Theorem~\ref{thm:worst_case} holds even if $R/I$ is not Cohen-Macaulay.

\begin{theorem}
For a generically chosen subspace $\U \subseteq \H$ of dimension $\dim(\U) \leq N-\dim(\X)$, it holds that $I_k+I(\U)_k=R_k$ for
\ba
k=\setft{reg}(R/I)+1.
\ea
\end{theorem}
\begin{proof}
Let $\dim(\X)\leq r \leq N$, let $\ell_1,\dots, \ell_r$ be generically chosen linear forms, and let $\U=V(\ell_1,\dots, \ell_r)$ so that $\U$ is a generically chosen subspace of dimension $N-r$. Let $A=R/I$. By~\cite[Lemma 4.9, Corollary 4.11]{eisenbud2005geometry}, a generically chosen linear form $\ell_1 \in R_1$ satisfies $\setft{reg}(A) \geq \setft{reg}(R/(I+\langle \ell_1 \rangle))$. Continuing in this way, we obtain $\setft{reg}(A) \geq \setft{reg}(R/(I+\langle \ell_1,\dots, \ell_r \rangle))$. Furthermore, $\U \cap \X=\{0\}$ (see e.g.~\cite[Definition 11.2]{harris2013algebraic}). %and $\setft{reg}(A) \geq \setft{reg}(R/(I+\langle \ell_1,\dots, \ell_r \rangle))$.
By~\cite[Corollary 4.4]{eisenbud2005geometry}, $\setft{reg}(R/(I+\langle \ell_1,\dots, \ell_r \rangle))$ is the largest integer $k$ for which $I_k+I(\U)_k \neq R_k$. Hence equality holds for $k=\reg(A)+1$.
\end{proof}

The following theorem gives a bound on the dimension of a generically chosen subspace that can be certified $\X$-tangled at level $k$ of the Nullstellensatz hierarchy.
%This result is stronger than the degree bounds presented in~\cite{JLV22a,JLV23a}, and is sharp.\footnote{The degree bound proven in~\cite{JLV23a} has the advantage of also applying to a related algorithm for \textit{recovering} elements of $\X$ contained in $\U$. It is an interesting open problem to adapt Theorem~\ref{thm:generic_intro} to this setting.}
\begin{theorem}\label{thm:generic}
Let $s,k$ be positive integers. If
\ba\label{eq:generic}
\dim I_k^\perp<{N-s+k\choose k},
\ea
then a generically chosen $s$-dimensional subspace $\U\subseteq \H$ satisfies $I_k+I(\U)_k=R_k$.
\end{theorem}
\begin{remark}\label{rmk:related}
Weaker results were presented in~\cite{JLV22a} and~\cite[Theorem 20]{JLV23a}, with similar guarantees for a related algorithm to \textit{recover} elements of $\X$ contained in $\U$. The above theorem improves~\cite[Theorem 20]{JLV23a} for the problem of certifying $\X$-tanglement of a subspace, but it does not apply to the recovery setting. It is an interesting open problem to adapt Theorem~\ref{thm:generic} to the recovery setting. This would amount to proving that for generically chosen $\psi_1 \psi_1^*,\dots, \psi_s \psi_s^* \in \X$ with $s$ satisfying~\eqref{eq:generic} (or similar), it holds that $I_k+I(\spn\{\psi_1,\dots, \psi_s\})_k=I(\psi_1,\dots, \psi_s)_k$.
\end{remark}
\begin{example}
Theorem~\ref{thm:generic} is sharp. Let $I=\langle x_1,x_2,\dots,x_{s-1} \rangle $ and let $\X=V(I)$.
Then we have $\dim I_k^\perp={N-s+k\choose k}$.
Any subspace $\U\subseteq \H$ of dimension $s$ must intersect $\X$. This implies
that $I_k+I(\U)_k$ cannot equal $R_k$.
\end{example}

To prove Theorem~\ref{thm:generic} we fix a monomial ordering and invoke Galligo's theorem that the leading monomial ideal of $I$ is Borel fixed after a change of basis. In more details, we fix the {\em reverse lexicographic} monomial ordering on $R$ with $x_1>x_2>\cdots>x_N$, wherein $x_1^{\alpha_1} \cdots x_N^{\alpha_N} > x_1^{\beta_1} \cdots x_N^{\beta_N}$ if and only if either
\begin{enumerate}
\item $|\bfalpha| > |\bfbeta|$, or
\item $|\bfalpha|=|\bfbeta|$, and $\alpha_i < \beta_i$ for the greatest integer $i$ satisfying $\alpha_i \neq \beta_i$.
\end{enumerate}
Note that for any positive integers $s$ and $k$, 
 any monomial of degree $k$ in $x_1,x_2,\dots,x_s$
 is larger than any monomial of degree $k$ in
 the ideal $\langle x_{s+1},x_{s+2},\dots,x_N\rangle$.
A \textit{monomial ideal} is an ideal generated by monomials. For an ideal $I\subseteq R$, the \textit{leading monomial ideal} of $I$, denoted $\lm(I)$, is the ideal generated by the leading monomials of the elements of $I$. For a polynomial $f$ let $\lm(f)$ be the leading monomial of $f$.

\begin{definition}
An monomial ideal $J\subseteq R$ is \emph{Borel-fixed} if for every monomial $q\in J$ and every $i$ and $j$ with $1\leq i<j\leq n$ we have that $x_j q\in I$ implies $x_i q\in I$.
\end{definition}
Note that if $x_s^k$ lies in a Borel-fixed ideal $J$,
then all monomials in $x_1,x_2,\dots,x_s$ of degree $k$ lie in $J$.

\begin{theorem}[Galligo~\cite{galligo1979theoreme}. See also Theorems~15.18 and~15.20 in~\cite{eisenbud2013commutative}]\label{thm:galligo}
If $I\subseteq R$ is a homogeneous ideal, then there exists a linear change of coordinates on $\H$ such that $\lm(I)$ is a Borel-fixed ideal.
\end{theorem}

Now we can prove Theorem~\ref{thm:generic}.

\begin{proof}[Proof of Theorem~\ref{thm:generic}]
Since $I_k+I(\U)_k=R_k$ is a Zariski open condition on the Grassmannian, it suffices to prove that it holds for a single $s$-dimensional subspace $\U$. By Theorem~\ref{thm:galligo}, after a linear change of coordinates we may assume that $J:=\lm(I)$ is Borel-fixed. 
Suppose that $x_s^k\not\in J$.
Then all the monomials of degree $k$ in $J$ lie in the ideal $\langle x_1,x_2,\dots,x_{s-1}\rangle$. It follows that
\ba
\dim I_k^\perp&=\dim J_k^\perp\\
&=\dim(R_k/J_k)\\
&\geq \dim (R_k/\langle x_1,\dots, x_{s-1} \rangle_k)\\
&=\dim \complex[x_s,x_{s+1},\dots,x_N]_k\\
&={N-s+k\choose k},
\ea
where the first line follows from $R_k/I_k \cong R_k/J_k$, which is a general property of leading monomial ideals due to Macaulay (see e.g.~\cite[Lemma 4.2.1]{schenck2003computational}), and the other lines are obvious. This contradicts our assumptions, so $x_s^k$ lies in $J$. Let $q_1>q_2>\cdots>q_\ell$ be all monomials of degree $k$ in $x_1,x_2,\dots,x_s$.
Because $J$ is Borel-fixed, $q_1,q_2,\dots,q_\ell\in J$. Let $\U=V(x_{s+1},\dots, x_N)$, and let $K=I(\U)$.
There exist homogeneous
polynomials $f_1,f_2,\dots,f_\ell\in I$ of degree $k$ such that $\lm(f_i)=q_i$ for all $i$.
If we express
$f_1+K_k,f_2+K_k,\dots,f_\ell+K_k$
in the basis
$q_1+K_k,q_2+K_k,\dots,q_\ell+K_k$ 
of
$R_k/K_k\cong \complex[x_1,x_2,\dots,x_s]_k$
we get a lower triangular matrix. This proves
that $f_1+K_k,f_2+K_k,\dots,f_\ell+K_k$ is also a basis, so $I_k+K_k=R_k$. This completes the proof.
\end{proof}

\begin{table}
\begin{center}
\renewcommand{\arraystretch}{1.5}
\begin{tabular}{ |c|c|c|c| } 
  \hline
  $\X=V(I)\subseteq \P(\H)$ & worst case degree & generic degree \\ 
  \hline
   \hline
  $\X$ any variety & $\leq N(d-1)+1$ & $\leq \min k$ s.t. \\ 
&& $\dim I_k^{\perp} < \binom{\eps N+k}{k}$\\
%  \hline
%  $\X_1 \subseteq \P(\complex^{n_1 \times n_2})$ & $\min\{n_1,n_2\}$ & $O(1)$ \\ 
  \hline
  $\X_r \subseteq \P(\complex^{n_1}\otimes \complex^{n_2})$ & $r(\min\{n_1,n_2\}-r)+1$ & $O_r(1)$ \\
  \hline 
  $\X_{\setft{Sep}}\subseteq \P(\complex^{n_1}\otimes \dots \otimes \complex^{n_m})$ & $\sum_{j=1}^{m} (n_j)-\max_j n_j-m+2$ & $O(1)$\\
%  \hline
%    $\X_{\setft{Bisep}}\subseteq \P(\complex^{n_1}\otimes \dots \otimes \complex^{n_m})$ & $\max_{\emptyset \neq S \subsetneq [m]} \min\{\prod_{j \in S} n_j, \prod_{j \in [m]\setminus S} n_j\}$& $O(1)$\\
    \hline
    $\X_{\ell\text{-}\Sep}\subseteq \P(\complex^{n_1}\otimes \dots \otimes \complex^{n_m})$ & $\max_{B\text{ an $\ell$-part. of $[m]$}} \min_i\{\prod_{j \in B_i} n_j\}$& $O(1)$\\
    \hline
    $\X_{t\text{-Prod}}\subseteq \P(\complex^{n_1}\otimes \dots \otimes \complex^{n_m})$ & $\max_{B\text{ part. of $[m]$, $|B_i|\geq t$}} \min_i\{\prod_{j \in B_i} n_j\}$& $O(1)$\\
    \hline
  $\X_{\setft{Sep}}^{\vee}\subseteq \P(S^m(\complex^n))$ & $n-\ceil{\frac{n}{m}}+1$ & $O_m(1)$\\
  \hline
  $\X_{\setft{Sep}}^{\w}\subseteq \P(\Lambda^m(\complex^n))$ & $\leq \binom{n}{m}+1$ & $O_m(1)$\\
  \hline
  $\X_{\setft{MPS}, r}\subseteq \P(\complex^{n_1}\otimes \dots \otimes \complex^{n_m})$ & $\leq n_1\cdots n_m\cdot r+1$ & $O_r(1)$\\
 \hline
 $\X_{\setft{PEPS},r}\subseteq \P((\complex^n)^{\otimes (m_1 m_2)})$ &  & $O_r(1)$\\
 \hline
 $\X_{\setft{G, r}}\subseteq \P(\complex^{n_1}\otimes \dots \otimes \complex^{n_m})$ &  & $O_{r^{\delta(G)}}(1)$\\
%  $\X_{\setft{TTNS,r}}\subseteq  \complex^{n_1}\otimes \dots \otimes \complex^{n_m}$ & $\leq n_1\cdots n_m\cdot r+1$ & $O_r(1)$\\
  \hline
\end{tabular}
\end{center}
\caption{For different choices of varieties $\X\subseteq \P(\H)$, the degree $k$ of the Nullstellensatz hierarchy needed to certify $\X$-tanglement of a linear subspace $\U\subseteq \H$ in the worst case and for $\U$ generically chosen of dimension $(1-\eps) \dim(\H)$, where $0<\eps<1$ is arbitrary but fixed. In the first row, $\X=V(I)\subseteq \P(\H)$ is an arbitrary variety with $I$ generated in degree at most $d$ and $\dim(\H)=N$. In the last two rows, $\X$ is not an algebraic variety, but is contained in $\X_{r'}$ for some $r'$ and choice of bipartition, so one can certify that $\U$ is $\X$-tangled by running the Nullstellensatz hierarchy for $\X_{r'}$. However, this will not certify \textit{every} $\X$-tangled subspace, so the worst-case degree is undefined.}\label{table:nullstellensatz}
\end{table}

\subsection{Examples}\label{sec:examples}

In this section we apply our worst-case and generic degree bounds for the Nullstellensatz hierarchy to the varieties in Example~\ref{ex:varieties}. Some of these bounds are summarized in Table~\ref{table:nullstellensatz}.

Recall the Schur representations of $\Unitary(\J)$, reviewed in Section~\ref{sec:rep}. For our analysis we will require a standard formula for $\dim(S^{\lambda}(\J))$, which we now review (see e.g.~\cite{landsberg2012tensors} or~\cite{fulton2013representation}). Let $T_{\lambda}$ be the \textit{Young diagram} of $\lambda$: A diagram of left-justified rows, with row $i$ having $\lambda_i$ boxes. For example, if $\lambda=(3,2,2)\vdash 7$, then $T_{\lambda}$ is given by

%\begin{center}
%\ytableausetup{textmode}
%\begin{ytableau}
% &  &  \\
% &  &  \\
%c
%\end{ytableau}
%\end{center}
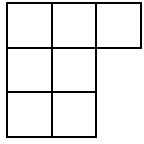
\begin{center}
\ydiagram{3,2,2}
\end{center}
If $\dim(\J)=n$, then
\ba\label{eq:schur_dim}
\dim(S^{\lambda}(\J))=\prod_{x \in T_{\lambda}} \frac{n+c(x)}{h(x)},
\ea
where $x$ ranges over the boxes in $T_{\lambda}$. Here, $c(x)$ is the \textit{content} of $x$, which is zero if $x$ is on the diagonal, $s$ if $x$ is $s$ steps above the diagonal, and $-s$ if $x$ is $s$ steps below the diagonal. $h(x)$ is the \textit{hook length} of $x$, which is the number of boxes to the right of $x$ plus the number of boxes below $x$ plus one.

%For a conic variety $\X$ whose only non-smooth point is the origin, it holds that $\complex[x]/I(\X)$ is Cohen-Macaulay (see e.g.~\cite[Chapter 18]{eisenbud2013commutative}. For $\X_r \subseteq \complex^{n_1} \otimes \complex^{n_2}$ the variety of matrices of rank at most $r$, it holds that $\complex[x]/I(\X_r)$ is Cohen-Macaulay. (This was first proven by Eagon and Hochster. See also~\cite[Theorem 18.18]{eisenbud2013commutative} and~\cite{bruns2006determinantal}.)

\begin{example}[Bipartite product states]\label{ex:schmidtrankone}
Let $\X_1 \subseteq \P(\complex^{n_1}\otimes \complex^{n_2})$ be the product pure states. Then $R/I(\X_1)$ is Cohen-Macaulay and has regularity $\min\{n_1,n_2\}-1$ (see e.g.~\cite[Section 1]{morales2016veronese}). Thus, by Theorem~\ref{thm:worst_case}, for a linear subspace $\U \subseteq \complex^{n_1} \otimes \complex^{n_2}$ the following statements are equivalent:
\begin{enumerate}
\item $\U$ is $\X_1$-tangled.
\item $I(\X_1)_{k} + I(\U)_{k} = R_{k}$ for $k=\min\{n_1,n_2\}$.
\end{enumerate}
and this bound on $k$ is tight. Let $0 < \eps < 1$ be fixed, and let $s=(1-\eps)n_1 n_2$. We now show that a generically chosen subspace of dimension $s$ is certified to avoid $\X_1$ at a constant level $k$. Since $\dim I(\X_1)_k^{\perp} = \binom{n_1+k-1}{k}\binom{n_2+k-1}{k}$, by Theorem~\ref{thm:generic} it suffices to prove that
\ba
\binom{n_1+k-1}{k}\binom{n_2+k-1}{k}<\binom{\eps n_1 n_2 +k}{k}
\ea
for $k=O(1)$. This is clear, since for all $n_1, n_2 \geq 2$ it holds that
\ba\label{eq:biseg}
\binom{n_1+k-1}{k}\binom{n_2+k-1}{k}&< \frac{e^{2k} (n_1+k-1)^k (n_2+k-1)^k}{k^{2k}}\\
%&\leq \frac{(n_1+k)^k (n_2+k)^k}{e^{2k(\ln(k)-1)}}\\
&\leq \frac{(\eps n_1 n_2+k)^k}{k^k}\\
&\leq \binom{\eps n_1 n_2 +k}{k},
\ea
where the second line holds for a suitably large constant $k$, and the other lines follow from the inequalities
\ba\label{eq:binom_bounds}
\frac{n^k}{k^k}\leq \binom{n}{k} < \frac{e^k n^k}{k^k}.
\ea
\end{example}

\begin{example}[Schmidt rank]\label{ex:sub_schmidt}
Let $\X_r \subseteq \P(\complex^{n_1}\otimes \complex^{n_2})$ be the set of pure states of Schmidt rank at most $r$. Then $R/I(\X_r)$ is Cohen-Macaulay and has regularity $r(\min\{n_1,n_2\}-r)$ (see~\cite{LASCOUX1978202},~\cite[Chapter 6]{Weyman_2003}, or~\cite[Example 4.3]{rajchgot2021degrees}). Thus, by Theorem~\ref{thm:worst_case}, for a linear subspace $\U \subseteq \complex^{n_1} \otimes \complex^{n_2}$ the following statements are equivalent:
\begin{enumerate}
\item $\U$ is $\X_r$-tangled.
\item $I(\X_r)_{k} + I(\U)_{k} = R_{k}$ for $k=r(\min\{n_1,n_2\}-r)+1$,
\end{enumerate}
and this bound is tight.

Let $0 < \eps < 1$ be arbitrary but fixed, and let $s=(1-\eps)n_1 n_2$. We now show that a generically chosen subspace of dimension $s$ is certified to avoid $\X_r$ at a constant level $k$. This follows from Theorem~\ref{thm:generic} and the fact that for all $n_1,n_2 > r$ we have
\ba
\dim I(\X_r)_k^{\perp}&=\sum_{\substack{\lambda \vdash k\\ \ell(\lambda)\leq r}} \dim(S^{\lambda}(\complex^{n_1})) \dim(S^{\lambda}(\complex^{n_2}))\\
				&\leq \sum_{\substack{\lambda \vdash k\\ \ell(\lambda)\leq r}} \frac{(n_1+k)^k(n_2+k)^k}{(\lambda_1! \cdots \lambda_r!)^2}\\
				&\leq \binom{k+r-1}{r-1} \frac{(n_1+k)^k(n_2+k)^k}{(\floor{\frac{k}{r}}!)^{2r}}\\
				&< \frac{e^r(k+r-1)^r}{(r-1)^{r-1}} \frac{(n_1+k)^k(n_2+k)^k}{(\floor{\frac{k}{r}}!)^{2r}}\\
				&\leq \frac{(\eps n_1 n_2+k)^k}{k^k}\\
				&\leq \binom{\eps n_1 n_2 +k}{k},\label{eq:schmidt}
\ea
for a suitably chosen constant $k=O_r(1)$ which may depend on $r$. The first line follows from~\eqref{eq:schmidt_rep}. The second line follows from
\ba\label{eq:schurdimbound}
\dim(S^{\lambda}(\complex^{n}))&=\prod_{x\in T_{\lambda}} \frac{n+c(x)}{h(x)}\\
						&\leq \frac{(n+k)^k}{\prod_{x\in T_{\lambda}} h(x)}\\
						&\leq \frac{(n+k)^k}{\lambda_1! \cdots \lambda_r!}.
\ea
The third line in~\eqref{eq:schmidt} follows from the fact that there are at most $\binom{k+r-1}{r-1}$ terms in the sum, and each is upper bounded by $\frac{(n_1+k)^k(n_2+k)^k}{(\floor{\frac{k}{r}}!)^{2r}}$. The fourth and sixth line follow from the inequalities~\eqref{eq:binom_bounds}. The fifth line holds for $k=O_r(1)$ large enough.
\end{example}

\begin{example}[Multipartite product states]\label{ex:multipartite_seg}
Let $\X_{\setft{Sep}} \subseteq \P(\complex^{n_1}\otimes \dots \otimes \complex^{n_m})$ be the set of pure product states. Then $R/I(\X_{\setft{Sep}})$ is Cohen-Macaulay and has regularity $\sum_{j=1}^{m} (n_j)-\max_j n_j-m+1$ (see e.g.~\cite[Section 1]{morales2016veronese}). By Theorem~\ref{thm:worst_case}, for a linear subspace $\U \subseteq \complex^{n_1} \otimes \dots \otimes \complex^{n_m},$ the following statements are equivalent:
\begin{enumerate}
\item $\U$ is $\X_{\setft{Sep}}$-tangled.
\item $I(\X_{\setft{Sep}})_{k} + I(\U)_{k} = R_{k}$ for $\sum_{j=1}^{m} (n_j)-\max_j n_j-m+2$,
\end{enumerate}
and this bound is tight.

Let $0 < \eps < 1$ be arbitrary but fixed, and let $s=(1-\eps)n_1 n_2 \cdots n_m$. We now show that a generically chosen subspace of dimension $s$ is certified to avoid $\X_{\setft{Sep}}$ at a constant level $k$. Let $\X_1\subseteq \P(\complex^{n_1}\otimes \dots \otimes \complex^{n_m})$ be the set of pure states of the form $\varppsi \otimes \varpphi$, where $\psi \in \complex^{n_1}$ and $\phi \in \complex^{n_2}\otimes \dots \otimes \complex^{n_m}$. It follows that
\ba
\dim(I(\X_{\setft{Sep}})_k^{\perp})&\leq \dim(I(\X_1)_k^{\perp})\\
&< \binom{\eps n_1 \cdots n_m+k}{k},
\ea
where the first line follows from $\X_{\setft{Sep}}\subseteq \X_{1}$ and the second line follows from~\eqref{eq:biseg}. This completes the proof by Theorem~\ref{thm:generic}.
\end{example}

\begin{example}[Biseparable states] Let $\X_{\setft{Bisep}}\subseteq \P(\complex^{n_1}\otimes \dots \otimes \complex^{n_m})$ be the set of biseparable states, those that have Schmidt rank one with respect to some bipartite cut. Given a subspace $\U \subseteq \complex^{n_1}\otimes \dots \otimes \complex^{n_m}$, rather than running the Nullstellensatz hierarchy directly for $\X_{\setft{Bisep}}$ it seems more economical to run it on the set of Schmidt rank one states $\X_1^{(S,T)}$
% \subseteq (\otimes_{i \in S} \complex^{n_i})\otimes (\otimes_{j \in T} \complex^{n_j})$ 
with respect to each non-trivial bipartition $S \sqcup T = [m]$, and conclude that $\U \cap \X_{\setft{Bisep}} = \{0\}$ if each of these hierarchies certifies that $\U \cap \X_1^{(S,T)} = \{0\}$. By Example~\ref{ex:schmidtrankone}, the following statements are equivalent:
\begin{enumerate}
\item $\U$ is $\X_{\setft{Bisep}}$-tangled.
\item For every non-trivial bipartition $S \sqcup T=[m]$, it holds that $I(\X_1^{(S,T)})_k+I(\U)_k=R_k$ for $k=\min\{\prod_{j \in S} n_j, \prod_{j \in [m]\setminus S} n_j\}$.
\end{enumerate}
By Example~\ref{ex:schmidtrankone} again, for any constant $0 < \eps < 1$, a generically chosen subspace of dimension $(1-\eps)n_1 \cdots n_m$ is certified to avoid $\X_{\setft{Bisep}}$ by running each hierarchy up to a constant level $k=O(1)$.

More generally, this analysis can be easily extended to $\ell$-separable states and $t$-producible states, with similar results that we record in Table~\ref{table:nullstellensatz}.
%   $\X_{\setft{Bisep}}\subseteq \complex^{n_1}\otimes \dots \otimes \complex^{n_m}$ & $\max_{\emptyset \neq S \subsetneq [m]} \min\{\prod_{j \in S} n_j, \prod_{j \in [m]\setminus S} n_j\}$& $O(1)$\\
\end{example}

\begin{example}[Bosonic product states]
Let $\X_{\setft{Sep}}^{\vee} \subseteq \P(S^m(\complex^{n}))$ be the set of pure symmetric product states. Then $R/I(\X_{\setft{Sep}}^{\vee})$ is Cohen-Macaulay and has regularity $n-\lceil{\frac{n}{m}}\rceil$ (see e.g.~\cite[Section 1]{morales2016veronese}). By Theorem~\ref{thm:worst_case}, for a linear subspace $\U \subseteq S^m(\complex^n),$ the following statements are equivalent:
\begin{enumerate}
\item $\U$ is $\X_{\setft{Sep}}^{\vee}$-tangled.
\item $I_{k} + I(\U)_{k} = R_{k}$ for $k=n-\ceil{\frac{n}{m}}+1$,
\end{enumerate}
and this bound is tight.

Let $0 < \eps < 1$ be arbitrary but fixed, and let $s=(1-\eps) \binom{n+m-1}{m}$. We now show that a generically chosen subspace of dimension $s$ is certified to avoid $\X_{\setft{Sep}}^{\vee}$ at a constant level $k=O_m(1)$ (which may depend on $m$). This follows from Theorem~\ref{thm:generic} and
\ba
\dim(I(\X_{\setft{Sep}}^{\vee})_k^{\perp}) &\leq \dim(I(\X_{\setft{Sep}})_k^{\perp})\\
								&\leq \binom{\eps n^m/m^m + k}{k}\\
								&\leq \binom{\eps \binom{n+m-1}{m}+k}{k},
\ea
%\ba
%\dim(I(\X_{\setft{Sep}}^{\vee})_k^{\perp}) &= \binom{n+mk-1}{mk}\\
%								&< \frac{e^{mk}(n+mk-1)^{mk}}{(mk)^{mk}}\\
%%								&\leq \left(\frac{\eps^m (n+m-1)^m}{m^m}+k\right)\\
%								&\leq \frac{(\eps(n+m-1)^m+k)^k}{m^{mk} k^k}\\
%								&\leq \binom{\eps \binom{n+m-1}{m}+k}{k},
%\ea
where the first line follows from $\X_{\setft{Sep}}^{\vee} \subseteq \X_{\setft{Sep}}$, the second line follows from Example~\ref{ex:multipartite_seg} with $\eps \rightarrow \eps/m^m$, and the third line follows from~\eqref{eq:binom_bounds}.
\end{example}

\begin{example}[Fermionic product states]
Let $\X_{\setft{Sep}}^{\w} \subseteq \P(\Lambda^m(\complex^{n}))$ be the set of unentangled Fermionic pure states. Then $R/I(\X_{\setft{Sep}}^{\w})$ is Cohen-Macaulay~\cite{HOCHSTER197340}. $I(\X_{\Sep}^{\w})$ is generated by the \textit{Plücker relations}, a set of degree-2 homogeneous polynomials described explicitly in~\cite[Eq. 3.4.10]{jacobson2009finite}. By Theorem~\ref{thm:worst_case_bound}, for a linear subspace $\U \subseteq \Lambda^m(\complex^n),$ the following statements are equivalent:
\begin{enumerate}
\item $\U \cap \X_{\setft{Sep}}^{\w} =\{0\}$.
\item $I(\X_{\setft{Sep}}^{\w})_{k} + I(\U)_{k} = R_{k}$ for $k=\binom{n}{m}+1$.
\end{enumerate}
Let $0 < \eps < 1$ be arbitrary but fixed, and let $s=(1-\eps) \binom{n}{m}$. We now show that a generically chosen subspace of dimension $s$ is certified to avoid $\X_{\setft{Sep}}^{\w}$ at a constant level $k=O_m(1)$ (which may depend on $m$).  This follows from Theorem~\ref{thm:generic} and the fact that for any fixed $m$ and any $n > m$ it holds that
\ba
\dim(I(\X_{\setft{Sep}}^{\w})_k^{\perp})&=\dim(S^{k^{(m)}}(\complex^n))\\
&\leq \frac{(n+mk)^{mk}}{(k!)^m}\\
&\leq \frac{\eps^k n^{mk}}{m^{mk}k^k}\\
&\leq \binom{\eps \binom{n}{m}+k}{k},
\ea
where the second line follows similarly as in~\eqref{eq:schurdimbound}, the second line holds for $k=O_m(1)$ large enough, and the third line follows from the inequalities~\eqref{eq:binom_bounds}.
\end{example}

\begin{example}[Matrix product states]
Let $\X_{\setft{MPS},r} \subseteq \P(\complex^{n_1}\otimes \dots \otimes \complex^{n_m})$ be the variety of matrix product states (MPS) of bond dimension $r$, which is cut out set-theoretically by a collection of homogeneous polynomials of degree $r+1$ (see Example~\ref{ex:mps}). Let $I$ be the ideal generated by these minors. By Theorem~\ref{thm:worst_case_bound}, for a linear subspace $\U \subseteq \complex^{n_1}\otimes \dots \otimes \complex^{n_m}$ the following statements are equivalent:
\begin{enumerate}
\item $\U$ is $\X_{\setft{MPS},r}$-tangled.
\item $I_{k} + I(\U)_{k} = R_{k}$ for $k=r n_1 \cdots n_m+1$.
\end{enumerate}
Let $0 < \eps < 1$ be arbitrary but fixed, and let $s=(1-\eps)n_1 n_2 \cdots n_m$. We now show that a generically chosen subspace of dimension $s$ is certified to avoid $\X_{\setft{MPS},r}$ at a constant level $k=O_r(1)$. Let $\X_r \subseteq \P(\complex^{n_1}\otimes \dots \otimes \complex^{n_m})$ be the set of pure states of Schmidt rank at most $r$ with respect to the bipartition $(\complex^{n_1})\otimes (\otimes_{j \neq 1} \complex^{n_j})$. It follows that
\ba
\dim(I_k^{\perp})&\leq \dim(I(\X_r)_k^{\perp})\\
&< \binom{\eps n_1 \cdots n_m+k}{k}
\ea
for $k=O_r(1)$, where the first line follows from $I_k \supseteq I(\X_{r})_k$, and the second line follows from~\eqref{eq:schmidt}. This completes the proof by Theorem~\ref{thm:generic}. Analogous results hold more generally for tree tensor network states (TTNS). Yet more generally, the following examples demonstrate that we can even say something for tensor network states with an arbitrary underlying graph.
\end{example}

\begin{example}[PEPS]
Let $\X_{\setft{PEPS}, r}^{\circ} \subseteq \P((\complex^n)^{\otimes m_1m_2})$ be the set of~\textit{Projected Entangled Pair States} (PEPS) of bond dimension $r$ on an $m_1 \times m_2$ rectangular lattice. This is the set of pure states that can be prepared by applying linear maps to the vertices of an $m_1 \times m_2$ lattice of pure states of Schmidt-rank $r$ (see~\cite{cirac2021matrix} for more details). Let $\X_{\setft{PEPS}, r}$ be the closure of $\X_{\setft{PEPS}, r}^{\circ}$ (Euclidean or Zariski, they are equivalent).

In contrast to most of the other examples, we do not know a set of equations cutting out $\X_{\setft{PEPS}, r}$ (see~\cite{de2022linear,bernardi2023dimension} for partial progress). However, we can still certify $\X_{\setft{PEPS}, r}$ of a subspace $\U$ by certifying $\Y$-tanglement of $\U$ for some variety $\Y\supseteq \X_{\setft{PEPS}, r}$ for which we do know the equations. Let $\Y_{r^2} \subseteq \P((\complex^n)^{\otimes m})$ be the set of pure states of Schmidt rank at most $r^2$ with respect to the bipartite cut between a corner vertex and the rest of the graph. Then $\Y_{r^2} \supseteq \X_{\setft{PEPS}, r}$ and we can certify $\Y_{r^2}$-tanglement of a generically chosen subspace of dimension $(1-\eps)n^{m_1 m_2}$ at level $k=O_{r}(1)$ of the Nullstellensatz hierarchy, by Example~\ref{ex:sub_schmidt}.

% for $k=O_{r}(1)$, where the first line follows from $\X_{\setft{PEPS},r}\subseteq \X_{r^2}$, and the second line follows from~\eqref{eq:schmidt}.
%
%Let $0 < \eps < 1$ be arbitrary but fixed, let $s=(1-\eps)n^{m_1 m_2}$, and let $\X_{ r^2 } \subseteq \P((\complex^n)^{\otimes m})$ be the variety of pure states of Schmidt rank at most $r^2 $ with respect to the bipartite cut between a corner vertex of the graph (say, the top left corner) and the other vertices. Then $\X_{\setft{PEPS}, r}\subseteq \X_{ r^2 }$, so we can certify $\X_{\setft{PEPS}, r}$-tanglement by

% We now show that a generically chosen subspace of dimension $s$ is certified to avoid $\X_{\setft{PEPS},r}$ at a constant level $k=O_{r}(1)$ of the Nullstellensatz hierarchy. Let $\X_{ r^2 } \subseteq \P((\complex^n)^{\otimes m})$ be the variety of pure states of Schmidt rank at most $r^2 $ with respect to the bipartite cut between a corner vertex of the graph (say, the top left corner) and the other vertices. Then for $n > r^2$, $\X_{r^2}$ is a proper subset of $\P((\complex^n)^{\otimes m})$ and it holds that
%\ba
%\dim(I(\X_{\setft{PEPS},r})_k^{\perp})&\leq \dim(I(\X_{r^2})_k^{\perp})\\
%&< \binom{\eps n_1 \cdots n_m+k}{k}
%\ea
%for $k=O_{r}(1)$, where the first line follows from $\X_{\setft{PEPS},r}\subseteq \X_{r^2}$, and the second line follows from~\eqref{eq:schmidt}. This completes the proof by Theorem~\ref{thm:generic}.
\end{example}

\begin{example}[Tensor network states]
For a graph $G$ on $m$ vertices and a positive integer $r$, let $\X_{G,r} \subseteq \P((\complex^{n})^{\otimes m})$ be the closure (Zariski or Euclidean; they are equivalent) of the set of tensor network states (TNS) of bond dimension $r$ with underlying graph $G$. This is a natural generalization of PEPS to an arbitrary underlying graph, with Schmidt rank $r$ states placed on the edges. See e.g.~\cite{bernardi2023dimension} for details.

%By Theorem~\ref{thm:worst_case_bound}, for a linear subspace $\U \subseteq \complex^{n_1}\otimes \dots \otimes \complex^{n_m}$ the following statements are equivalent:
%\begin{enumerate}
%\item $\U \cap \X =\{0\}$.
%\item $I_{k} + I(\U)_{k} = \complex[x]_{k}$ for $k=r n_1 \cdots n_m+1$.
%\end{enumerate}

As for PEPS, we do not know a set of equations cutting out $\X_{G,r}$. Let $\delta(G)$ be the minimum vertex degree of $G$, and let $\Y_{ r^{\delta(G)}} \subseteq \P((\complex^n)^{\otimes m})$ of pure states of Schmidt rank at most $r^{\delta(G)}$ with respect to the bipartite cut between a minimum degree vertex and the other vertices. Then $\Y_{ r^{\delta(G)}} \supseteq \X_{G,r}$, and 
we can certify $\Y_{ r^{\delta(G)}}$-tanglement of a generically chosen subspace of dimension ${(1-\eps)n^{|G|}}$ at level $k=O_{r}(1)$ of the Nullstellensatz hierarchy, by Example~\ref{ex:sub_schmidt}.

%Let $0 < \eps < 1$ be arbitrary but fixed, and let $s=(1-\eps)n^m$. We now show that a generically chosen subspace of dimension $s$ is certified to avoid $\X_{G,r}$ at a constant level $k=O_{r^{\delta(G)}}(1)$ of the Nullstellensatz hierarchy, where $\delta(G)$ is the minimum vertex degree of $G$. Let $\X_{ r^{\delta(G)}} \subseteq \P((\complex^n)^{\otimes m})$ be the variety of tensors of Schmidt rank at most $r^{\delta(G)}$ with respect to the bipartite cut between a minimum degree vertex and the other vertices. For $n > r^{\delta(G)}$, $\X_{r^{\delta(G)}}$ is a proper subset of $\P((\complex^n)^{\otimes m})$. It follows that
%\ba
%\dim(I(\X_{G,r})_k^{\perp})&\leq \dim(I(\X_{r^{\delta(G)}})_k^{\perp})\\
%&< \binom{\eps n_1 \cdots n_m+k}{k}
%\ea
%for $k=O_{r^{\delta(G)}}(1)$, where the first line follows from $\X_{G,r}\subseteq \X_{r^{\delta(G)}}$, and the second line follows from~\eqref{eq:schmidt}. This completes the proof by Theorem~\ref{thm:generic}.
\end{example}

\subsection{An equivalence of hierarchies}\label{sec:equiv}

An alternative hierarchy for determining if $\U$ is $\X$-tangled can be obtained from Corollary~\ref{cor:subspace}, which in particular says that $\U$ is $\X$-tangled if and only $\nu_k <1$ for $k \gg 0$. Satisfyingly, this is equivalent to the Nullstellensatz hierarchy:

%
%\begin{definition}[Nullstellensatz hierarchy for determining subspace $\X$-tanglement]
%Let $\field$ be a field, let $I \subseteq \field[x]$ be a homogeneous ideal, and let $\U \subseteq \field^n$ be a linear subspace. At level $k$, the \textit{Nullstellensatz hierarchy} checks if $I_k+I(\U)_k=\field[x]_k$. If equality holds, then $\U\cap \X=\{0\}$. If $\field=\complex$ (or more generally if $\field$ is algebraically closed), then this hierarchy is \textit{complete} in the sense that $\U \cap \X=\{0\}$ if and only if $I_k+I(\U)_k=\field[x]_k$ for some positive integer $k$.
%\end{definition}

%The last three authors studied this hierarchy (using different notation) in the works~\cite{JLV22a,JLV22b}.

%The nullstellensatz hierarchy is equivalent to the version of our hierarchy that simply checks whether $\nu_k <1$:
%we have presented in Proposition~\ref{prop:subspace} is \textit{equivalent} to this hierarchy when specialized to only checking whether $\U$ is $\X$-tangled:
\begin{prop}
It holds that $\nu_k<1 \iff I_k + I(\U)_k= S^k(\H^*)$.
\end{prop}
\begin{proof}
We have
\ba
\nu_k<1\iff &\im(\Pi_{\U} \otimes \I_{\H}^{\otimes k-1}) \cap \im(\Pi_{I,k}) = \{0\}\\
\iff & (\U \otimes \H^{\otimes k-1}) \cap I_k^{\perp} = \{0\}\\
\iff & (\U \otimes \H^{\otimes k-1}) \cap S^k(\H) \cap I_k^{\perp} = \{0\}\\
\iff & S^k(\U) \cap I_k^{\perp}=\{0\}\\
\iff & I(\U)_k^\perp \cap I_k^{\perp} =\{0\}\\
\iff & I(\U)_k+I_k=S^k(\H^*),
\ea
where the first two lines are obvious, the third line follows from $I_k^{\perp} \subseteq S^k(\H)$, the fourth line follows from $(\U \otimes (\H)^{\otimes k-1}) \cap S^k(\H)= S^k(\U)$, the fifth line follows from $S^k(\U)=I(\U)_k^{\perp}$, and the last line follows from taking the orthogonal complement with respect to the bilinear pairing.
\end{proof}

We can also relate the Nullstellensatz hierarchy to the hierarchy of linear systems introduced in~\cite{JLV22a}. There, in the examples they considered, the authors defined a map $\Phi\in \End(\H^{\otimes k})$ for which $\ker(\Phi) \cap S^k(\H)= I_k^{\perp}$ (see also Example~\ref{ex:schmidt}). They then took a basis $u_1,\dots, u_{\ell}$ for $S^k(\U)$, and checked if $\{\Phi(u_1),\dots, \Phi(u_{\ell})\}$ is linearly independent. This holds if and only if $I_k^{\perp} \cap \U^{\otimes k}=\{0\}$, which in turn holds if and only if $I_k + I(\U)_k=R_k$. Note that this amounts to solving a linear system of size $\dim(I_k) \times \binom{\dim(\U)+k-1}{k}$.

\section{Hermitian polynomial optimization}\label{sec:poly}

In this section we use Theorem~\ref{thm:robust_hier} to obtain a hierarchy of eigencomputations for optimizing a Hermitian form over a variety. We then prove that this hierarchy is equivalent to the Hermitian sum-of squares (HSOS) hierarchy of~\cite[Theorem 2.1]{d2009polynomial} (see also~\cite[Section 2.1]{wang2023real}) applied to this setting. While HSOS is in general a hierarchy of semidefinite programs, our results show that it is equivalent to a hierarchy of eigencomputations in our setting. This can be advantageous, as eigencomputations can be computed much faster in practice than full semidefinite programs.

\subsection{Background on Hermitian polynomials}
%Let $\V$ be a finite-dimensional complex Hilbert space of dimension $n$ with positive definite inner product $\ip{\cdot}{\cdot}$ which we take to be antilinear in the first argument. The inner product defines an antilinear isomorphism $\V \cong \V^*$ given by $v^*=\ip{v}{-}$. We extend this to an inner product on the symmetric algebra $S^{\bullet}(\V)$ by $\ip{u^d}{v^d}=\ip{u}{v}^d$.
%We extend this to an inner product on the tensor algebra $\V^{\bullet}$ by $\ip{u_1\otimes \dots \otimes u_d}{v_1 \otimes \dots \otimes v_d}=\ip{u_1}{v_1}\cdots \ip{u_d}{v_d}$.
We begin with some background on Hermitian polynomials. See e.g.~\cite{d2019hermitian} for more details. For a homomorphism $R \in \Hom(S^c(\H),S^d(\H))$, let $R(z,w^*):=\ip{w^{\otimes d}}{R z^{\otimes c}}$. More generally, for $R \in \Hom(S^{\uc}(\H),S^{\ud}(\H))$ with degree $(a,b)$ block $R_{(a,b)} \in \Hom(S^a(\H),S^b(\H))$, let $R(z,w^*)=\sum_{a,b} R_{(a,b)}(z,w^*)$.
%\ba
%R=\sum_{\substack{a \leq c\\ b \leq d}} R_{(a,b)}
%\ea
%for $R_{(a,b)} \in \Hom(S^a(\H),S^b(\H))$, let $R(z,w^*)=\sum_{a,b} R_{(a,b)}(z,w^*)$.
% We view will tensors in $S^{c}(\H^*)\otimes S^{d}(\H)$ as polynomials in $\H$ and $\H^*$ of bi-degree $(c,d)$, and also as homomorphisms $S^{c}(\H^*)\otimes S^{d}(\H)\cong \Hom(S^{c}(\H),S^{d}(\H)).$ More generally, we view $S^{\bullet}(\H^*)\otimes S^{\bullet}(\H)\cong \End(S^{\bullet}(\H)).$ 
The following proposition is standard, and shows that $R \in \Hom(S^{\ud}(\H),S^{\uc}(\H))$ is uniquely determined by $R(z,z^*)$. See also \cite{d2011hermitian,d2019hermitian}.

\begin{prop}\label{prop:polarization}
Let $R \in \Hom(S^{\uc}(\H),S^{\ud}(\H))$. If $R(z,z^*)=0$ for all $z \in \H$, then $R=0$.
\end{prop}
\begin{proof}
Let $(C_{\bfalpha,\bfbeta})_{|\bfalpha|\leq c, |\bfbeta|\leq d}$ be the matrix of $R$ in the monomial basis. Then $R(z,w^*)=\sum_{|\bfalpha|\leq c, |\bfbeta|\leq d} C_{\bfalpha, \bfbeta} z^{\bfalpha} w^{*\bfbeta}$. For $z=(e^{i \theta_1} x_1,\dots, e^{i\theta_N} x_N) \in \H$ with $x_i$ real we have
\ba
R(z,z^*)=\sum_{\bfalpha,\bfbeta} C_{\bfalpha,\bfbeta} e^{i (\alpha_1-\beta_1)\theta_1} \cdots e^{i (\alpha_N-\beta_N)\theta_N} x_1^{\alpha_1+\beta_1}\cdots x_N^{\alpha_N+\beta_N}=0.
\ea
So for any $\bfgamma \in \integer^N$ it holds that
\ba
0&=\int_{\theta_N=0}^{2\pi}\cdots \int_{\theta_1=0}^{2\pi} R(z,z^*) e^{i \gamma_1 \theta_1} \cdots e^{i \gamma_N \theta_N} d\theta_1 \dots d\theta_N\\
&=\sum_{\bfbeta-\bfalpha=\bfgamma} C_{\bfalpha,\bfbeta} x^{\bfalpha+\bfbeta}.
\ea
Since $C_{\bfalpha,\bfbeta} x^{\bfalpha+\bfbeta}$ is the only term of degree $\bfalpha+\bfbeta$ in the sum, it follows that $C_{\bfalpha,\bfbeta}=0$, so $R=0$.
%
%
%\ba
%R(x,x^*)=R(x,x)=\sum_{\alpha, \beta} c_{\alpha, \beta} x^{\alpha+\beta}=0
%\ea
%as a real polynomial, so for each index tuple $\gamma$ it holds that $\sum_{\alpha+\beta=\gamma} c_{\alpha,\beta}=0$. Now, letting $z(\theta)=(e^{i \theta_1},\dots, e^{i\theta_n})$ we have
%\ba
%R(z({\theta}),z({\theta})^*):=\sum_{\alpha,\beta} c_{\alpha,\beta} e^{i (\alpha_1-\beta_1)\theta_1} \cdots e^{i (\alpha_n-\beta_n)\theta_n}
%\ea
%is identically zero as a function of $\theta$. It follows that for any index tuple $\gamma$, we have
%\ba
%\sum_{\alpha-\beta=\gamma} c_{\alpha,\beta}= \int_{\theta_n=0}^{2\pi}\cdots \int_{\theta_1=0}^{2\pi} R(z(\theta),z(\theta)^*) e^{i (\beta_1-\alpha_1)\theta_1} \cdots e^{i (\beta_n-\alpha_n)\theta_n} d\theta_1 \dots d\theta_n =0.
%\ea 
%so $R=0$. This completes the proof. 
\end{proof}

A \textit{Hermitian polynomial} is an element $H \in \End(S^{\ud}(\H))$ for which $H(z,w^*)=\overline{H(w,z^*)}$ for all $z,w \in \H$. Let $\Herm(S^{\ud}(\H))\subseteq \End(S^{\ud}(\H))$ be the set of Hermitian polynomials, and let $\Herm(S^{d}(\H))\subseteq \End(S^{d}(\H))$ be the set of Hermitian forms of bi-degree $(d,d)$. An immediate consequence of Proposition~\ref{prop:polarization} is that $H(z,z^*)=\sum_{\bfalpha,\bfbeta} C_{\bfalpha,\bfbeta} z^{\bfalpha} z^{*\bfbeta}$ is Hermitian if and only if it is Hermitian as an endomorphism, i.e. the matrix $(C_{\bfalpha,\bfbeta})_{\bfalpha,\bfbeta}$ is Hermitian (see also~\cite[Proposition 1.1]{d2011hermitian}). Let $\Pos(S^{\ud}(\H))$ be the set of \textit{Hermitian sums-of-squares}: Hermitian polynomials of the form $P(z,z^*)=\sum_{i} |q_i(z)|^2$ for some $q_i \in S^{\ud}(\H^*)$. Similarly, these are the polynomials that are positive semidefinite as endomorphisms, i.e. the matrix $(C_{\bfalpha,\bfbeta})_{\bfalpha,\bfbeta}$ is positive semidefinite.

Note that the Hermitian form $H_k=H(z,z^*) \norm{z}^{2(k-d)}$ as an operator is given by $H_k=\Pi_k(H \otimes \I_{\H}^{\otimes{(k-d)}})\Pi_k$, where $\Pi_k$ is the orthogonal projection onto $S^k(\H)$, since $\ip{z^{\otimes k}}{H_k  z^{\otimes k}}=H(z,z^*) \norm{z}^{2(k-d)}$ under this choice. As before, for a homogeneous ideal $I \subseteq S^{\bullet}(\H^*)$, let $I_k^{\perp} \subseteq S^k(\H)$ be the orthogonal complement to $I_k$ with respect to the bilinear pairing $\H \times \H^* \rightarrow \complex$ and let $\Pi_{I,k} \in \Pos(S^k(\H))$ be the orthogonal projection onto $I_k^{\perp}$. Recall that we define $V(I)=\{vv^* : f(v)=0 \;\; \forall f \in I\}\subseteq \P(\H)$.

%Given a Hermitian form $H \in S^d(\H^*)$ and a homogeneous ideal $I \subseteq S^{\bull}(\H^*)$, we consider optimizing $H$ over the variety cut out by $I$, i.e. computing the quantity
%\ba
%\min_{\substack{zz^* \in V(I)}} H(z,z^*).
%\ea

\subsection{A hierarchy of eigencomputations}
Theorem~\ref{thm:robust_hier} implies the following hierarchy for optimizing a Hermitian form subject to homogeneous equality constraints. In the following, let $\lambda_{\min}(\cdot)$ be the minimum eigenvalue.
\begin{theorem}\label{thm:hermitian}
Let $H \in \Herm(S^d(\H))$ be a Hermitian form and let $I \subseteq S^{\bullet}(\H^*)$ be a homogeneous ideal. For each $k\geq d$ let $H_k=H(z,z^*) \norm{z}^{2(k-d)} \in \Herm(S^k(\H))$ and $\nu_k=\lambda_{\min}(\Pi_{I,k} H_k \Pi_{I,k})$. Then $\nu_d,\nu_{d+1},\dots $ forms a non-decreasing sequence for which
\ba
\lim_{k\rightarrow \infty} \nu_k = \min_{\substack{zz^* \in V(I)}} H(z,z^*).
\ea
\end{theorem}
\begin{proof}
The fact that the sequence $\nu_d,\nu_{d+1},\dots$ is non-decreasing can be easily shown using similar techniques as in the proof of Theorem~\ref{thm:robust_hier}.

Let $J \subseteq S^{\bullet}(S^d(\H^*))$ be the homogeneous ideal with $J_{\ell}=I_{\ell d} + S^{\ell d}(\H)^{\perp} \subseteq S^{\ell}(S^d(\H^*))$ for each $\ell$. Then $V(J)=\{z^{\otimes d} : zz^* \in V(I)\}$, and we have
%
%Let $\ell$ be such that $V(I)=V(I_{\ell d})$, and let $J$ be the ideal in $S^{\bullet}(S^d(\H^*))$ generated by $I_{\ell d} \subseteq S^{\ell}(S^d(\H^*))$. Then
\ba
\min_{\substack{ zz^* \in V(I)}} H(z,z^*)&=\min_{\substack{ zz^* \in V(I)}} \ip{z^{\otimes d}}{H z^{\otimes d}}\\
&=\min_{\substack{ \psi\psi^* \in V(J)}} \ip{\psi}{H \psi}.
\ea
By Theorem~\ref{thm:robust_hier}, this is equal to $\lim_{k\rightarrow \infty} \mu_{dk}$, where we define $\mu_{dk}=\lambda_{\min}(\Pi_{J,k} H_{dk} \Pi_{J,k})$. Note that $J_k^{\perp}=I_{dk}^{\perp} \cap S^{dk}(\H)$, so $\Pi_{J,k}$ is the projection onto the image of $\Pi_{I,dk}$. It follows that $\mu_{dk}=\nu_{dk}=\lambda_{\min}(\Pi_{I,dk} H_{dk} \Pi_{I,dk})$. So the sequence $\nu_d,\nu_{d+1},\dots $ has the same limit as the sequence $\mu_{d},\mu_{2d},\dots$ and the statement holds.
\end{proof}
\subsection{Relationship to HSOS}
Theorem~\ref{thm:hermitian} is closely related to the HSOS hierarchy for Hermitian polynomial optimization~\cite[Theorem 2.1]{d2009polynomial} (see also~\cite[Section 2.1]{wang2023real}). For a subset $S \subseteq S^{\ud}(\H^*)$, let $|S|^2 \subseteq \Pos(S^{\leq 2d}(\H))$ be the set of polynomials that can be written as sums of Hermitian squares $|f(z)|^2$ of elements $f\in S$. Note that $f \in S^{\ud}(\H^*)$ satisfies $f(z)=0$ if and only if $|f(z)|^2\leq 0$. Using this to translate equality constraints into inequality constraints, and invoking~\cite[Lemma 1]{lyubich2005polynomial}, the result~\cite[Theorem 2.1]{d2009polynomial} gives the following hierarchy for Hermitian optimization under equality constraints:
\begin{theorem}\label{thm:hsos}
Let $H \in \Herm(S^{\ud}(\H))$ and let $I \subseteq S^{\bullet}(\H^*)$ be a homogeneous ideal. If $H(z,z^*)>0$ for all $zz^* \in V(I)$, then there exists $k \in \natural$ for which
\ba\label{eq:hsos}
H(z,z^*)\in \Pos(S^{\underline{k}}(\H))-|I_{\underline{k}}|^2+(1-\norm{z}^2) \Herm(S^{\underline{k}}(\H)).
\ea
%= P(z,z^*)- |f_1(z)|^2 Q_1(z,z^*) - \dots - |f_p(z)|^2 Q_p(z,z^*)+(1-\norm{z}^2)A(z,z^*)
%\ea
%for some $P, Q_1,\dots, Q_p\in \Pos(S^{\underline{k}}(\H))$ and $A \in \Herm(S^{\underline{k}}(\H))$.
\end{theorem}

% In the case that we consider, this hierarchy says if $H(z,z^*) > 0$ for all $z \in V(I)$, then $H(z,z^*)= P(z,z^*)+ f_1(z) A_1(z,z^*) + \dots + f_p(z) A_p(z,z^*)$ for some $P \in \Pos(S^{\underline{k}}(\H))$ and $A_1,\dots, A_p \in \Herm(S^{\underline{k}}(\H))$.
%\begin{theorem}[Theorem 2.1 of~\cite{d2009polynomial}]
%Let $H, Q_1,\dots, Q_{\ell} \in \Herm(S^{\bullet}(\H))$. Then
%\end{theorem}
%

We now prove that the hierarchies described in Theorems~\ref{thm:hermitian} and~\ref{thm:hsos} are equivalent when $H$ is bihomogeneous, in the sense that $\nu_k >0$ if and only if~\eqref{eq:hsos} holds. The proof is similar in spirit to \cite[Proposition 2]{de2005equivalence}.

\begin{theorem}
Let $I \subseteq S^{\bullet}(\H^*)$ be a homogeneous ideal, let $H \in \Herm(S^{d}(\H)),$ and let $H_k=H(z,z^*) \norm{z}^{2(k-d)} \in \Herm(S^k(\H))$. The following statements are equivalent:
\begin{enumerate}
%\item $\Pi_{I,k} (H \otimes \I_{\H}^{\otimes{(k-d)}}) \Pi_{I,k}\in \Pos(S^k(\H))$
\item $H(z,z^*)\in \Pos(S^{\underline{k}}(\H))-|I_{\underline{k}}|^2+(1-\norm{z}^2) \Herm(S^{\underline{k}}(\H)).$
\item $H_k \in \Pos(S^{{k}}(\H)) - |I_{{k}}|^2$.
\item $\Pi_{I,k} H_k \Pi_{I,k} \in \Pos(S^k(\H))$.
\end{enumerate}
\end{theorem}
\begin{proof}
The equivalence $(2\iff 3)$ is straightforward: Statement 3 holds if and only if $H_k =P-Q$ for some $P,Q \in \Pos(S^k(\H))$ with $\im(Q)\subseteq I_k^{\perp}$. But the set of such $Q$ is precisely $|I_k|^2$.

$(2 \Rightarrow 1)$:  Statement $2$ implies that
\ba
H_k&=H\cdot (1-1+\norm{z}^2)^{k-d} \\
&=H+H\sum_{s=1}^{k-d} \binom{k-d}{s} (\norm{z}^2-1)^s \in \Pos(S^{{k}}(\H)) - |I_{{k}}|^2.
\ea
Subtracting terms containing nonzero powers of $(\norm{z}^2-1)$ from both sides, we obtain
\ba
H\in \Pos(S^{{k}}(\H)) - |I_{{k}}|^2 + (1-\norm{z}^2)\Herm(S^k(\H)).
\ea

$(1 \Rightarrow 2)$: Let $P \in \Pos(S^{\underline{k}}(\H)), Q \in |I_{\underline{k}}|^2, S \in \Herm(S^{\underline{k}}(\H))$ be such that
\ba
H=P-Q+(1-\norm{z}^2)S.
\ea
Setting $\hat{z}=z/\norm{z}$, we have
\ba
H(\hat{z},\hat{z}^*)=P(\hat{z},\hat{z}^*)-Q(\hat{z},\hat{z}^*).
\ea
%We first observe that we can assume neither $P$ nor $Q$ contain any terms of degree $(c,c')$ for $c \neq c'$. Indeed, since $H$ is bihomogeneous it is phase invariant $H(e^{i \theta} {z},e^{-i \theta} {z}^*)=H({z},{z}^*)$, so $P-Q$ is also phase invariant on unit vector inputs. Let $T\in \Hom(S^c(\H),S^{c'}(\H))$ be the homogeneous component of $P-Q$ of degree $(c,c')$. It follows from Proposition~\ref{prop:polarization} that $T$ is phase invariant on unit vector inputs. But if $c \neq c'$ then this implies $T(\hat{z},\hat{z}^*)=T(e^{i \theta} \hat{z},e^{-i \theta} \hat{z}^*)=e^{i\theta(c-c')} T(\hat{z},\hat{z}^*)$, so $T=0$ by Proposition~\ref{prop:polarization}.
%
%Redefine $P,Q$ by removing all terms of degree $(c,c')$ with $c \neq c'$. This leaves $P-Q$ unchanged. If we consider $P,Q$ as block matrices with blocks in $\Hom(S^c(\H),S^{c'}(\H))$, then we are zeroing out the off-diagonal blocks. Clearly ${P} \in \Pos(S^{\underline k}(\H))$ still holds, and since $I$ is homogeneous, ${Q} \in |I_{\underline k}|^2$ still holds.
Thus,
\ba
H_k=H(\hat{z},\hat{z}^*)\norm{z}^{2k} = P(\hat{z},\hat{z}^*)\norm{z}^{2k} - Q(\hat{z},\hat{z}^*) \norm{z}^{2k}.
\ea
Since $H_k$ is bihomogeneous of degree $(k,k)$, every other component must cancel, so
\ba
H_k(z,z^*)=P_{(d,d)}(z,z^*) \norm{z}^{2(k-d)}-Q_{(d,d)}(z,z^*)\norm{z}^{2(k-d)}
\ea
where $P_{(d,d)}$ is the component of $P$ of bi-degree $(d,d)$, and similarly for $Q$. Clearly $P_{(d,d)} \norm{z}^{2(k-d)} \in \Pos(S^k(\H))$, and since $I$ is homogeneous, $Q_{(d,d)}\in  |I_{ d}|^2$, so ${Q}_{(d,d)} \norm{z}^{2(k-d)} \in |I_{ k}|^2$. This completes the proof.
% follows that $H(e^{i \theta} \hat{z},e^{-i \theta} \hat{z}^*)=H(\hat{z},\hat{z}^*)$ for all $\theta \in \real$. So $P($
%We complete the proof by showing that $P':=P(\hat{z},\hat{z}^*) \norm{z}^{2k} \in \Pos(S^k(\H))$ and $Q':=Q(\hat{z},\hat{z}^*) \norm{z}^{2k} \in |I_k|^2$. 
\end{proof}

\section{Acknowledgments}
We thank Aravindan Vijayaraghavan for valuable insights in the early stages of this work. We thank Sujit Rao for helpful discussions. H.D. was supported by the National Science Foundation under Grant No. DMS-2147769. N.J. was supported by NSERC Discovery Grant RGPIN2022-04098. B.L. acknowledges that this material is based upon work supported by the National Science Foundation under Award No. DMS-2202782.
%% A.V. was supported by the National Science Foundation under Grant Nos. CCF-1652491, CCF-1934931, ECCS-2216970. %, and a Google Research Scholar award.

%\paragraph{Data availability statement} This manuscript has no associated data.
%
%\paragraph{Declarations} The research leading to these results received funding from the National Science Foundation under Grant Nos. DMS-2147769, DMS-2202782 and NSERC Discovery Grant RGPIN2022-04098.

%-----------------------------------------------------------------------------%
\bibliographystyle{alpha}
\bibliography{references}
%-----------------------------------------------------------------------------%

\appendix

\end{document}